%% file: flexible-rates.tex
\documentclass[11pt,letterpaper,english]{article}

\usepackage{amsmath,amssymb,amsthm}
\usepackage[ruled,vlined]{algorithm2e}
\usepackage{pgf, tikz}
\usetikzlibrary{arrows}
\usepackage{xspace, units}
\usepackage{typearea}
\allowdisplaybreaks

\usepackage{lmodern}
\usepackage[T1]{fontenc}
\usepackage{textcomp}

\typearea{16}
                             
\usepackage{color}
\definecolor{Darkblue}{rgb}{0,0,0.4}
\definecolor{Brown}{cmyk}{0,0.81,1.,0.60}
\definecolor{Purple}{cmyk}{0.45,0.86,0,0}
\usepackage[breaklinks]{hyperref}
\hypersetup{colorlinks=true,
            citebordercolor={.6 .6 .6},linkbordercolor={.6 .6 .6},
            citecolor=black,urlcolor=Darkblue,linkcolor=black,
            pdfpagemode=UseThumbs,pdfstartview=FitH}
\newcommand{\lref}[2][]{\hyperref[#2]{#1~\ref*{#2}}}

\newtheorem{theorem}{Theorem}

\newtheorem{lemma}[theorem]{Lemma}
\newtheorem{corollary}[theorem]{Corollary}

\newtheorem{claim}[theorem]{Claim}
\newtheorem{observation}[theorem]{Observation}

\newcommand{\RR}{\ensuremath{\mathbb{R}}}

\newcommand{\R}{{\mathcal R}}
\renewcommand{\L}{{\mathcal L}}

\renewcommand{\Pr}[1]{\mbox{\rm\bf Pr}\left[#1\right]}
\newcommand{\Ex}[1]{\mbox{\rm\bf E}\left[#1\right]}

\newcommand{\noise}{\ensuremath{N}}

\DeclareMathOperator*{\OPT}{OPT}
\DeclareMathOperator*{\ALG}{ALG}

\author{Thomas Kesselheim\thanks{Department of Computer Science, RWTH Aachen University, Germany. This work has been supported by the UMIC Research
Centre, RWTH Aachen University.}}

\title{Approximation Algorithms for Wireless Link Scheduling\\ with Flexible Data Rates}

\begin{document}

\maketitle

\begin{abstract}
We consider scheduling problems in wireless networks with respect to flexible data rates. That is, more or less data can be transmitted per time depending on the signal quality, which is determined by the signal-to-interference-plus-noise ratio (SINR). Each wireless link has a utility function mapping SINR values to the respective data rates. We have to decide which transmissions are performed simultaneously and (depending on the problem variant) also which transmission powers are used.

In the capacity-maximization problem, one strives to maximize the overall network throughput, i.e., the summed utility of all links. For arbitrary utility functions (not necessarily continuous ones), we present an $O(\log n)$-approximation when having $n$ communication requests. This algorithm is built on a constant-factor approximation for the special case of the respective problem where utility functions only consist of a single step. In other words, each link has an individual threshold and we aim at maximizing the number of links whose threshold is satisfied. On the way, this improves the result in [Kesselheim, SODA 2011] by not only extending it to individual thresholds but also showing a constant approximation factor independent of assumptions on the underlying metric space or the network parameters.

In addition, we consider the latency-minimization problem. Here, each link has a demand, e.g., representing an amount of data. We have to compute a schedule of shortest possible length such that for each link the demand is fulfilled, that is the overall summed utility (or data transferred) is at least as large as its demand. Based on the capacity-maximization algorithm, we show an $O(\log^2 n)$-approximation for this problem.
\end{abstract}

\input{introduction}
\input{problemstatements}
\input{unlimitedpowers}
\input{fixedpowers}
\input{limitedpowers}
\input{algorithmwithrates}
\input{latencyminimization}
\input{thresholdssmallerthanone}

\input{discussion}

\bibliographystyle{plain}
\bibliography{bibliography,mypapers}

\begin{appendix}
\input{omittedproofs}
\end{appendix}
\end{document}

%% file: introduction.tex
\section{Introduction}
The performance of wireless communication is mainly limited by the fact that simultaneous transmissions interfere. In the presence of interference, the connection quality can deteriorate or successful communication can even become impossible. In engineering, the common way to measure the quality of a wireless connection is the signal-to-interference-plus-noise ratio (SINR). This is the ratio of the strength of the intended signal and the sum of all other signal strengths plus ambient noise. The higher it is, the better different symbols in a transmission can be distinguished and therefore the more information can be transmitted per time. The significance of the SINR in this matter is information-theoretically supported by the Shannon-Hartley theorem stating a theoretical upper bound on the information rate proportional to $\log(1+\text{SINR})$.

Existing algorithmic research on wireless networks mainly focused on transmissions using fixed data rates. This way successful reception becomes a binary choice: Either the quality is high enough and the transmission is successfully received or nothing is received due to too much interference. Interference constraints are, for example, given by thresholds on the SINR or by a graph whose edges model mutual conflicts. Typical optimization problems are similar to independent set or coloring problems. That is, the most important aspect is a discrete choice, making the problems usually non-convex and NP-hard.

This threshold assumption does not reflect the ability of wireless devices to adapt their data rates to different interference conditions. For example, a file transfer can be slowed down on a poor-quality connection but still be carried out. For this reason, a different perspective has been taken in some works on power control. Here, it is assumed that devices are perfectly able to adapt to the current conditions. This is reflected in a utility function, which is assumed to be a concave and differentiable function of the SINR. Under these assumptions, one can design distributed protocols for selecting transmission powers that converge to local optima of the summed utility \cite{Huang2006}. Applying further restrictions on the utility functions, the summed utility becomes equivalent to a convex optimization problem and can thus be solved in polynomial time \cite{Chiang2007}. These algorithms solve a purely continuous optimization problem as they select transmission powers for each link from a continuous set and do not make a decision which transmissions are carried out. Requiring a certain minimal SINR for each transmission to be successful is thus only possible if the resulting set of feasible solutions is not empty. Furthermore, only under relatively strict assumptions, one is able to find a global optimum. It is even not clear if it is appropriate to assume continuity of utility functions as today's standards only support a fixed number of data rates. 

In this paper, we take an approach that generalizes both perspectives. We assume each communication request has a utility function that depends arbitrarily on the SINR. That is, we are given a function for each transmission mapping each SINR to the throughput that can be achieved under these conditions. By selecting which transmissions are carried out and which powers are used, we maximize the overall network throughput. As we allow for non-continuous utility functions, this generalizes both the threshold and the convex objective functions. We only assume that the utility functions and their inverse can be evaluated efficiently. 

More formally, we consider the following \emph{capacity-maximization problem with flexible data rates}. One is given $n$ communication requests, each being a link, that is, a pair of a sender and a receiver. For each link $\ell$ we are given a utility function $u_{\ell}$, quantifying the value that each possible SINR $\gamma_{\ell}$ has. The task is to select a subset of these links and possibly transmission powers with the objective of maximizing $\sum_{\ell} u_{\ell}(\gamma_{\ell})$, that is, the overall network's throughput.

\subsection{Our Contribution}
We present $O(\log n)$-approximations for the capacity-maximization problem with flexible data rates, for the case that our algorithm has to specify the transmission powers as well as for a given power assignment. In both cases, the algorithm is built on one for the respective capacity-maximization problem with individual thresholds. Here, utility functions are step functions with $u_{\ell}(\gamma_{\ell}) = 1$ for $\gamma_{\ell} \geq \beta_{\ell}$ and $0$ otherwise. That is, we assume that each link has an individual threshold $\beta_{\ell}$ and we maximize the number of links whose threshold is satisfied. For this special case, we present constant-factor approximations for both variable and fixed transmission powers.

For the case of variable transmission powers, this extends the result in \cite{Kesselheim2011} in two ways. On the one hand, the algorithm in \cite{Kesselheim2011} relies on the fact that all thresholds are equal. In order to guarantee feasible solutions under individual thresholds, the links have to be processed in a different order and further checks have to be introduced. Furthermore, in contrast to \cite{Kesselheim2011}, we are able to show a constant approximation factor without any further assumptions on the metric space respectively the model parameters.

For fixed transmission powers, we extend a greedy algorithm by Halld\'orsson and Mitra~\cite{Halldorsson2011} that works with monotone, (sub-) linear power assignments. By modifying the processing order, we are able to prove a constant approximation factor independent of the thresholds. Furthermore, we present a simplified analysis.

In addition to capacity-maximization problems, we also consider latency minimization. That is, for each link there is a fixed amount of data that has to be transmitted in shortest possible time. We use the capacity-maximization algorithms repeatedly with appropriately modified utility functions. This way, we can achieve an $O(\log^2 n)$-approximation for both variants.

\subsection{Related Work}
In a seminal work on power control, Foschini and Miljanic~\cite{Foschini1992} give a very simple distributed algorithm for finding a power assignment that satisfies the SINR targets of all links. They show that it converges from any starting point under the assumption that the set of feasible solutions is not empty. In subsequent works, more sophisticated techniques have been presented (for an overview see \cite{singhkumar10}). Very recently, this problem has also been considered from an algorithmic point of view \cite{Lotker2011,Dams2011PowerControl} deriving bounds on how the network size or parameters determine the convergence time. While in these problems typically no objective function is considered, Huang et al.~\cite{Huang2006} present a game-theoretic approach to maximize the sum of link utilities, where the utility of each link is an increasing and strictly concave function of the SINR. They show that their algorithm converges to local optima of the sum-objective function. Chiang et al.~\cite{Chiang2007} present an approach to compute the global optimum for certain objective functions by the means of geometric programming in a centralized way. 
All these algorithms have in common that they solve a continuous optimization problem. That is, transmission powers are chosen from a continuous set in order to maximize a continuous function. When requiring a minimum SINR, the mentioned algorithms only work under the assumption that all links can achieve this minimum SINR simultaneously. Since this is in general not true, one may have to solve an additional scheduling problem. Many heuristics have been presented for these scheduling problems but recently a number of approximation algorithms were studied as well. Most of them assume a common, constant threshold $\beta$ for all links. Usually, for the approximation factors this $\beta$ is considered constant, which is not appropriate in our case.

For example, in a number of independent papers, the problem of finding a maximum feasible set under uniform transmission powers has been tackled. For example, Goussevskaia et al.~\cite{Goussevskaia2009} present an algorithm that computes a set that is at most a constant factor smaller than the optimal one under uniform powers. In contrast to this, Andrews and Dinitz~\cite{Andrews2009} compare the set they compute to the optimal one using an arbitrary power assignment. Their approximation factor is $O(\log \Delta)$, where $\Delta$ is the ratio between the largest and the smallest distance between a sender and its corresponding receiver. This bound is tight as uniform power assignments, in general, cannot achieve better results. This is different for square-root power assignments \cite{Fanghaenel2009a,Halldorsson2009a}, which choose powers proportional to $\sqrt{d^\alpha}$ for a sender-receiver pair of distance $d$. The best bound so far by Halld\'orsson and Mitra \cite{Halldorsson2011} shows that one can achieve an $O(\log \log \Delta + \log n)$-approximation this way. However, for large values of $\Delta$ the approximation factors can get as bad as $\Omega(n)$. In general, it is better to choose transmission powers depending on the selected link set. In \cite{Kesselheim2011,Wan2011} a constant-factor approximation for the combined problem of scheduling and power control was presented. While the analysis in the mentioned papers only proves this approximation factor for \emph{fading metrics} (i.e. $\alpha$ is greater than the doubling dimension than the metric), we show in this paper that this result actually holds for all metric spaces. 

Apart from these capacity-maximization problems and the online variants \cite{OnlineSPAA}, one has focused on latency minimization, that is scheduling all transmission requests within shortest possible time. For this problem, distributed ALOHA-like algorithms have been analyzed \cite{Fanghaenel2009,Kesselheim2010,Halldorsson2011ICALP}.

The mentioned results only consider that signals propagate deterministically, neglecting short-term effects such as scattering, which are typically modeled stochastically. Dams et al.~\cite{Dams2012Rayleigh} present a black-box transformation to transfer algorithmic results to Rayleigh-fading conditions. Plugging in the results in this paper, we get an $O(\log n \cdot \log^\ast n)$-approximation for the flexible-rate problems and an $O(\log^\ast n)$-approximation for the respective threshold variants.

The only approximation algorithm for individual thresholds was given by Halld\'orsson and Mitra \cite{HalldorssonINFOCOM2012}. It is a constant-factor approximation for capacity maximization for uniform transmission powers. Santi et al.~\cite{Santi2009} consider latency minimization with flexible data rates. This approach, however, only considers quite restricted utility functions that have to be the same for all links. Furthermore, they only consider uniform transmission powers.  

%% file: problemstatements.tex
\section{Formal Problem Statements}
We identify the network devices by a set of nodes $V$ in a metric space. If some sender transmits at a power level $p$ then this signal is received at a strength of $p / d^\alpha$ by nodes whose distance to the sender is $d$. The constant $\alpha$ is called path-loss exponent. Given a set $\L \subseteq V \times V$ and a power assignment $p\colon \L \to \RR_{\geq 0}$, the SINR of link $\ell = (s, r) \in \L$ is given by 
\[
\gamma_{\ell}(\L, p) = \frac{\frac{p(\ell)}{d(s, r)^\alpha}}{\sum_{\substack{\ell' = (s', r') \in \L, \ell' \neq \ell}} \frac{p(\ell')}{d(s', r)^\alpha} + \noise } \enspace.
\]
Here, $\noise$ denotes the constant ambient noise. To avoid ambiguities such as divisions by $0$, we assume it to be strictly larger than $0$. However, it may be arbitrarily small.

In the \emph{capacity-maximization problem with flexible data rates}, we are given a set $\R \subseteq V \times V$ of pairs of nodes of a metric space. For each link $\ell \in \R$, we are given a utility function $u_{\ell}\colon [0, \infty) \to \RR_{\geq 0}$. Furthermore, we are given a maximum transmission power $p_{\max} \in \RR_{> 0} \cup \{ \infty \}$. The task is to select a subset $\L$ of $\R$ and a power assignment $p\colon \L \to [0, p_{\max}]$ such that $\sum_{\ell \in \L} u_{\ell}(\gamma_{\ell}(\L, p))$ is maximized.

We do not require the utility functions $u_{\ell}$ to be continuous. It neither has to be represented explicitly. We only assume that two possible queries on the utilities can be carried out. On the one hand, we need to access the maximum utility $u_{\ell}(p_{\max} / \noise)$ of a single link $\ell$. Note that this value could be infinite if $p_{\max} = \infty$. We ignore this case as the optimal solution is not well-defined in this case. On the other hand, we assume that for each link $\ell$ given a value $B$ no larger than its maximum utility, we can determine the smallest SINR $\gamma_\ell$ such that $u_{\ell}(\gamma_\ell) \geq B$. Both kinds of queries can be carried out in polynomial time for common cases of utility functions such as logarithmic functions or explicitly given step functions. As a technical limitation, we assume that $u_{\ell}(\gamma_\ell) = 0$ for $\gamma_\ell < 1$. This is not a weakness of our analysis but rather of the approaches studied so far. In Section~\ref{sec:thresholdssmallerthanone} we show that neither greedy nor ALOHA-like algorithms can dig the potential of SINR values smaller than $1$.

In addition to this problem, we also consider the variant with fixed transmission powers. That is, we are given a set $\R \subseteq V \times V$ and a power assignment $p\colon \R \to \RR_{\geq 0}$. We have to select a set $\L \subseteq \R$ maximizing $\sum_{\ell \in \L} u_{\ell}(\gamma_{\ell}(\L, p))$. We assume that the respective assumptions on the utility functions apply.

%% file: unlimitedpowers.tex
\section{Capacity Maximization with Thresholds and Unlimited Powers}
As a first step towards the final algorithm, we consider the following simplified problem. We are given a set $\R$ of links and a threshold $\beta(\ell) \geq 1$ for each $\ell \in \R$. We have to find a subset $\L \subseteq \R$ of maximum cardinality and a power assignment $p\colon \L \to \RR_{\geq 0}$ such that $\gamma_\ell \geq \beta(\ell)$ for all $\ell \in \L$. In \cite{Kesselheim2011} an algorithm was presented that solves the special case where all $\beta(\ell)$ are identical.  

The algorithm in \cite{Kesselheim2011} inherently requires that all thresholds are equal. To guarantee feasibility of the solution, it makes use of the fact that each link chooses a transmission power proportional to the one needed to overcome interference from longer ones. In the case of identical thresholds, selecting the links therefore boils down to guaranteeing small mutual distances. In consequence, some fundamental changes in the algorithm are necessary, that will be presented in this section.

We iterate over the links ordered by their \emph{sensitivity} $\beta(\ell) \noise d(\ell)^\alpha$, which is the minimum power necessary to overcome ambient noise and to have $\gamma_\ell \geq \beta(\ell)$ in the absence of interference.

More formally, let $\pi$ be the ordering of the links by decreasing values of $\beta(\ell) d(\ell)^\alpha$ with ties broken arbitrarily. That is, if $\pi(\ell) < \pi(\ell')$ then $\beta(\ell) d(\ell)^\alpha \geq \beta(\ell') d(\ell')^\alpha$. Based on this ordering, define the following (directed) weight between two links $\ell$ and $\ell'$. If $\pi(\ell) > \pi(\ell')$, we set
\[
w(\ell, \ell') = \min \left\{ 1, \beta(\ell) \beta(\ell') \frac{d(s,r)^\alpha d(s', r')^\alpha}{d(s, r')^\alpha d(s', r)^\alpha} + \beta(\ell) \frac{d(s, r)^\alpha}{d(s, r')^\alpha} + \beta(\ell) \frac{d(s, r)^\alpha}{d(s', r)^\alpha} \right\} \enspace, 
\]
otherwise we set $w(\ell, \ell') = 0$. For notational convenience we furthermore set $w(\ell, \ell) = 0$ for all $\ell \in \R$.

Our algorithm now works as follows: It iterates over all links in order of decreasing $\pi$ values, i.e., going from small to large values of $\beta(\ell) d(\ell)^\alpha$. It adds $\ell'$ to the set $\L$ if $\sum_{\ell \in \L} w(\ell, \ell') \leq \tau$ for $\tau = \nicefrac{1}{6 \cdot 3^\alpha + 2}$.

Afterwards, powers are assigned iterating over all links in order of increasing $\pi$ values, i.e., going from large to small values of $\beta(\ell) d(\ell)^\alpha$. The power assigned to link $\ell'$ is set to 
\[
p(\ell') = 2 \beta(\ell') \noise d(s', r')^\alpha + 2 \beta(\ell') \sum_{\substack{\ell = (s, r) \in \L \\ \pi(\ell) < \pi(\ell')}} \frac{p(\ell)}{d(s, r')^\alpha} d(s', r')^\alpha\enspace.
\]
If there were only links of smaller $\pi$ values, this power would yield an SINR of exactly $2 \beta(\ell)$. One can show that due to the greedy selection condition, also taking the other links into account the SINR of link $\ell$ is at least $\beta(\ell)$.

\subsection{Feasibility}
In order to show feasibility, we have to adapt the respective proof in \cite{Kesselheim2011}. The most important difference is as follows. Each link has an indirect effect on itself due to the fact that links of smaller sensitivity adapt their powers to it. Due to the individual thresholds this effect can only be bounded when considering weights that are modified in the described way. 

\begin{theorem}
\label{theorem:unlimitedpower:feasibility}
Let $(\L, p)$ be the solution returned by the algorithm. Then we have $\gamma_\ell(\L, p) \geq \beta(\ell)$ for all $\ell \in \L$.
\end{theorem}

\begin{proof}
Let $\L = \{ \ell_1, \ldots, \ell_{\bar{n}} \}$ with $\pi(\ell_1) < \pi(\ell_2) < \ldots \pi(\ell_{\bar{n}})$. Furthermore, let $\ell_i = (s_i, r_i)$ and $\beta_i = \beta(\ell_i)$ for all $i \in [\bar{n}]$. In this notation, for each link $\ell_i$ the power is set to
\[
p_i = 2 \beta_i \noise d(s_i, r_i)^\alpha + 2 \cdot \sum_{j=1}^{i-1} p_{i, j} \enspace, \qquad \text{where } \qquad p_{i,j}=\frac{\beta_i p_j d(s_i, r_i)^\alpha}{d(s_j, r_i)^\alpha} \enspace.
\]

So $p_{i, j}$ is the adaptation of link $i$ due to the interference of a link of larger sensitivity link $j$. This indirect effect of a link can be bounded by its direct effect as follows.

\begin{observation}
\label{observation:indirecteffect}
For $i, k \in [{\bar{n}}] := \{1, \ldots, \bar{n}\}$, we have
\[
\sum_{j=\max\{i, k\} + 1}^{\bar{n}} \frac{p_{j,k}}{d(s_j, r_i)^\alpha} \leq 2 \cdot 3^\alpha \cdot \tau \cdot \frac{p_k}{d(s_k, r_i)^\alpha} \enspace.
\]
\end{observation}

\begin{proof}
Let $m = \max\{i, k\} + 1$. We have
\[
\sum_{j=m}^{\bar{n}} \frac{p_{j,k}}{d(s_j, r_i)^\alpha} = \sum_{j=m}^{\bar{n}} \frac{\beta_j p_k \cdot d(s_j, r_j)^\alpha}{d(s_j, r_i)^\alpha \cdot d(s_k, r_j)^\alpha} \enspace. 
\]

We split up the terms into two parts, namely $M_1 = \{ j \in [{\bar{n}}] \mid j \geq m, d(s_k, r_i) \leq 3 d(s_k, r_j) \}$ and $M_2 = \{ j \in [{\bar{n}}] \mid j \geq m, d(s_k, r_i) > 3 d(s_k, r_j) \}$.

For all $j \in M_1$, we have $d(s_k, r_j) \geq \nicefrac{1}{3} \cdot d(s_k, r_i)$. This yields
\[
\sum_{j \in M_1} \frac{\beta_j p_k \cdot d(s_j, r_j)^\alpha}{d(s_j, r_i)^\alpha \cdot d(s_k, r_j)^\alpha} \leq \frac{3^\alpha \cdot p_k}{d(s_k, r_i)^\alpha} \sum_{j \in M_1} \beta_j \frac{d(s_j, r_j)^\alpha}{d(s_j, r_i)^\alpha} \leq \frac{3^\alpha \cdot p_k}{d(s_k, r_i)^\alpha} \cdot \tau
\]

For all $j \in M_2$, we have by triangle inequality
\[
d(s_k, r_i) \leq d(s_k, r_j) + d(s_j, r_j) + d(s_j, r_i) \leq \nicefrac{1}{3} \cdot d(s_k, r_i) + 2 d(s_j, r_i) \enspace,
\]
where the last step is due to the definition of $M_2$ and the fact that $d(s_j, r_j) \leq d(s_j, r_i)$ because $\sum_{j \in [\bar{n}]} w(\ell_j, \ell_i) \leq \tau$ and the fact that $\beta_j \geq 1$. This implies $d(s_j, r_i) \geq \nicefrac{1}{3} \cdot d(s_k, r_i)$ yielding
\[
\sum_{j \in M_2} \frac{\beta_j p_k \cdot d(s_j, r_j)^\alpha}{d(s_j, r_i)^\alpha \cdot d(s_k, r_j)^\alpha} \leq \frac{3^\alpha p_k}{d(s_k, r_i)^\alpha} \sum_{j \in M_2} \beta_j \frac{d(s_j, r_j)^\alpha}{d(s_k, r_j)^\alpha} \leq \frac{3^\alpha \cdot p_k}{d(s_k, r_i)^\alpha} \cdot \tau 
\]

Altogether this yields the claim.
\end{proof}

Now, let us consider some fixed $i \in [\bar{n}]$. We need to show $\gamma_{\ell_i} \geq \beta_i$. We define
\[
S = \frac{p_i}{d(s_i, r_i)^\alpha} \quad I_> = \sum_{j=1}^{i-1} \frac{p_j}{d(s_j, r_i)^\alpha} \quad I_< = \sum_{j=i+1}^{\bar{n}} \frac{p_j}{d(s_j, r_i)^\alpha}  
\]
In this notation, we have $\gamma_{\ell_i} = \frac{S}{I_> + I_< + \noise}$. Furthermore, we chose the powers such that $S = 2 \beta_i (I_> + \noise)$. So, it remains to bound $I_<$.  Plugging in the definitions, we get
\[
I_< = \sum_{j=i+1}^{\bar{n}} \frac{p_j}{d(s_j, r_i)^\alpha} = \sum_{j=i+1}^{\bar{n}} \left( 2 \beta_j \noise \frac{d(s_j, r_j)^\alpha}{d(s_j, r_i)^\alpha} + 2 \sum_{k=1}^{j-1} \frac{p_{j, k}}{d(s_j, r_i)^\alpha} \right)
\]
By re-arranging the sums, this is equal to
\begin{align*}
& 2 \noise \sum_{j=i+1}^{\bar{n}} \beta_j \frac{d(s_j, r_j)^\alpha}{d(s_j, r_i)^\alpha} + 2 \sum_{j=i+1}^{\bar{n}} \sum_{k=1}^{i-1} \frac{p_{j, k}}{d(s_j, r_i)^\alpha} 
 + 2 \sum_{j=i+1}^{\bar{n}} \frac{p_{j, i}}{d(s_j, r_i)^\alpha} + 2 \sum_{j=i+1}^{\bar{n}} \sum_{k=i+1}^{j-1} \frac{p_{j, k}}{d(s_j, r_i)^\alpha} \\
= \;& 2 \noise \sum_{j=i+1}^{\bar{n}} \beta_j \frac{d(s_j, r_j)^\alpha}{d(s_j, r_i)^\alpha}  + 2 \sum_{k=1}^{i - 1} \sum_{j=i+1}^{\bar{n}} \frac{p_{j,k}}{d(s_j, r_i)^\alpha} 
 + 2 \sum_{j=i+1}^{\bar{n}} \frac{p_{j, i}}{d(s_j, r_i)^\alpha} + 2 \sum_{k=i+1}^{\bar{n}} \sum_{j=k+1}^{\bar{n}} \frac{p_{j,k}}{d(s_j, r_i)^\alpha} \enspace.
\end{align*}
For the first part of the sum we can use $\sum_{j \in [\bar{n}]} w(\ell_j, \ell_i) \leq \tau$ for the latter ones Observation~\ref{observation:indirecteffect} to see this is at most
\[
2 \noise \tau + 2 \cdot 3^\alpha \cdot \tau 2 \sum_{k=1}^{i-1} \frac{p_k}{d(s_k, r_i)^\alpha} + \frac{2 \tau}{\beta_i} \frac{p_i}{d(s_i, r_i)^\alpha} + 2 \cdot 3^\alpha \cdot \tau 2 \sum_{k=i+1}^{\bar{n}} \frac{p_k}{d(s_k, r_i)^\alpha} \enspace.
\]
In the first two terms of the sum, we can recognize the definition of $p_i$. Furthermore, the last term is again the definition of $I_<$ multiplied by $2 \cdot 3^\alpha \cdot \tau \cdot 2$. 
\[
I_< \leq \left( 2 \cdot 3^\alpha + 2 \right) \cdot \frac{\tau}{\beta_i} \frac{p_i}{d(s_i, r_i)^\alpha} + 2 \cdot 3^\alpha \cdot \tau 2 \cdot I_< \enspace.
\]
Replacing the definition of $S$, we get 
\[
I_< \leq \left( 2 \cdot 3^\alpha + 2 \right) \cdot \frac{\tau}{\beta_i} S + 2 \cdot 3^\alpha \cdot \tau 2 \cdot I_< \enspace.
\]
By definition of $\tau$, this yields $I_< \leq \frac{1}{2 \beta_i} S$. Using $S = 2 \beta_i (I_> + \noise)$, we can conclude that $\beta_i(I_> + I_< + \noise) \leq S$. This exactly means that $\gamma_i \geq \beta_i$.
\end{proof}

\subsection{Approximation Factor}
In this section, we show that the algorithm achieves a constant approximation factor. In contrast to the analysis in \cite{Kesselheim2011} our analysis does not rely on any assumptions on the parameter $\alpha$. We only require that the nodes are located in a metric space.

The central result our analysis builds on is a characterization of \emph{admissible sets}. A set of links $\L$ is called admissible if there is some power assignment $p$ such that $\gamma_\ell(\L, p) \geq \beta(\ell)$ for all $\ell \in \L$.
\begin{lemma}
\label{lemma:admissiblesetcharacterization}
For each admissible set $\L$ there is a subset $\L' \subseteq \L$ with $\lvert \L' \rvert = \Omega(\lvert \L \rvert)$ and $\sum_{\ell' \in \L'} w(\ell, \ell') = O(1)$ for all $\ell \in \R$.
\end{lemma}
That is, for each admissible set $\L$ there is a subset having the following property. If we take some further link $\ell$, that does not necessarily belong to $\L$, the outgoing weight of this link to all links in the subset is bounded by a constant. Before coming to the proof of this lemma, let us first show how the approximation factor can be derived.

\begin{theorem}
\label{theorem:unlimitedpowerapproximationfactor}
The algorithm is a constant-factor approximation.
\end{theorem}

\begin{proof}
Let $\ALG$ be the set of links selected by our algorithm. Let furthermore be $\OPT$ the set of links in the optimal solution and $\L' \subseteq \OPT$ be the subset described in Lemma~\ref{lemma:admissiblesetcharacterization}. That is, we have $\sum_{\ell' \in \L'} w(\ell, \ell') \leq c$ for all $\ell \in \R$ for some suitable constant $c$. It now suffices to prove that $\lvert \ALG \rvert \geq \frac{\tau}{c} \cdot \lvert \L' \setminus \ALG \rvert$.

All $\ell' \in \L' \setminus \ALG$ were not selected by the algorithm since the greedy condition was violated. That is, we have $\sum_{\ell \in \ALG} w(\ell, \ell') > \tau$. Taking the sum over all $\ell' \in \L' \setminus \ALG$, we get $\sum_{\ell' \in \L' \setminus \ALG} \sum_{\ell \in \ALG} w(\ell, \ell') > \tau \cdot \lvert \L' \setminus \ALG \rvert$.

Furthermore, by our definition of $\L'$, we have $\sum_{\ell \in \ALG} \sum_{\ell' \in \L'} w(\ell, \ell') \leq c \cdot \lvert \ALG \rvert$. In combination that yields
\[
\lvert \ALG \rvert \geq \frac{1}{c} \sum_{\ell \in \ALG} \sum_{\ell' \in \L'} w(\ell, \ell') \geq \frac{1}{c} \sum_{\ell' \in \L' \setminus \ALG} \sum_{\ell \in \ALG} w(\ell, \ell') > \frac{\tau}{c} \cdot \lvert \L' \setminus \ALG \rvert \enspace.
\]
\end{proof}

\subsection{Proof of Lemma~\ref{lemma:admissiblesetcharacterization} (Outline)}
For the sake of clarity, the formal proof of Lemma~\ref{lemma:admissiblesetcharacterization} has been shifted to the appendix. Instead, we present the major steps here, highlighting the general ideas.

In the first step, we use the fact that we can scale the thresholds by constant factors at the cost of decreasing the size of $\L$ by a constant factor. Considering such an appropriately scaled set $\L' \subseteq \L$, we use the SINR constraints and the triangle inequality to show 
\[
\frac{1}{\lvert \L' \rvert} \sum_{\ell = (s, r) \in \L'} \sum_{\substack{\ell' = (s', r') \in \L' \\ \pi(\ell') > \pi(\ell)}} \frac{\beta(\ell') d(s', r')^\alpha}{d(s', r)^\alpha} = O(1) \enspace.
\]
In the next step, we use the following property of admissible set. After reversing the links (i.e. swapping senders and receivers) a subset of a constant fraction of the links is also admissible. Combining this insight with the above bound and Markov inequality, we can show that for each admissible set $\L$ there is a subset $\L' \subseteq \L$ with $\lvert \L' \rvert = \Omega(\lvert \L \rvert)$ and 
\[
\sum_{\ell' = (s', r') \in \L'} \frac{\min\{ \beta(\ell) d(s, r)^\alpha, \beta(\ell') d(s', r')^\alpha \} }{\min\{d(s', r)^\alpha, d(s, r')^\alpha\}} = O(1) \qquad \text{ for all $\ell = (s, r) \in \L'$.}
\]
This subset has the property that not only for $\ell \in \L'$ but for all $\ell = (s, r) \in \R$ 
\[
\sum_{\substack{\ell' = (s', r') \in \L' \\ \pi(\ell') < \pi(\ell)}} \min \left\{ 1, \frac{\beta(\ell) d(s, r)^\alpha}{d(s', r)^\alpha} \right\} + \min \left\{ 1, \frac{\beta(\ell) d(s, r)^\alpha}{d(s, r')^\alpha} \right\} = O(1) \enspace.
\]
In the last step, we show that there is also a subset $\L'$ with $\lvert \L' \rvert = \Omega(\lvert \L \rvert)$ and
\[
\sum_{\substack{\ell' = (s', r') \in \L' \\ \pi(\ell') < \pi(\ell)}} \min \left\{ 1, \beta(\ell) \beta(\ell') \frac{d(s, r)^\alpha d(s', r')^\alpha}{d(s', r)^\alpha d(s, r')^\alpha} \right\} = O(1) \qquad \text{ for all $\ell = (s, r) \in \R$} \enspace.
\]
This is shown by decomposing the set $\L$. For one part we can use the result from above. For the remaining links we can show that among these links there has to be an exponential growth of the sensitivity and the distance to $\ell$. In combination, the quotients added up in the sum decay exponentially, allowing us to bound the sum by a geometric series.

%% file: fixedpowers.tex
\section{Capacity Maximization with Thresholds and Fixed Powers}

In this section, we consider the case that a power assignment $p\colon \R \to \RR_{\geq 0}$ is given. We assume that this power assignment is (sub-) linear and monotone in the link sensitivity. In our case of individual thresholds this means that if for two links $\ell, \ell' \in \R$ we have $\beta(\ell) d(s, r)^\alpha \leq \beta(\ell') d(s', r')^\alpha$, then $p(\ell) \leq p(\ell')$ and $\frac{1}{\beta(\ell)} \frac{p(\ell)}{d(s,r)^\alpha} \geq \frac{1}{\beta(\ell')} \frac{p(\ell')}{d(s',r')^\alpha}$. This condition is fulfilled in particular if all transmission powers are the same or if they are chosen by a (sub-) linear, monotone function depending on the sensitivity, e.g., linear power assignments ($p(\ell) \sim \beta(\ell) \cdot d(\ell)^\alpha$) and square-root power assignments ($p(\ell) \sim \sqrt{\beta(\ell) \cdot d(\ell)^\alpha}$).

For the case of identical thresholds, Halld\'orsson and Mitra \cite{Halldorsson2011} present a constant-factor approximation that can be naturally extended as follows. Given two links $\ell = (s, r)$ and $\ell' = (s', r')$, and a power assignment $p$, we define the \emph{affectance} of $\ell$ on $\ell'$ by
$a_p(\ell, \ell') = \min \left\{ 1, \beta(\ell') \frac{p(\ell)}{d(s, r')^\alpha} \Big/ \left( \frac{p(\ell')}{d(s', r')^\alpha} - \beta(\ell') \noise \right) \right\}$.
Algorithm~\ref{alg:capacitymaximizationfixedpowers} iterates over all $\ell' \in \R$ and adds a link to the tentative solution if the incoming and outgoing affectance to the previously selected links is at most $\nicefrac{1}{2}$. At the end, only the feasible links are returned.

\begin{algorithm}
initialize $\L = \emptyset$ \;
\For{$\ell' \in \R$ in decreasing order of $\pi$ values}{
\If{$\sum_{\ell \in \L} a_p(\ell, \ell') + a_p(\ell', \ell) \leq \frac{1}{2}$}{
add $\ell'$ to $\L$\;
}
}
return $\{ \ell' \in \L \mid \sum_{\ell \in \L} a_p(\ell, \ell') < 1 \}$\;
\caption{Capacity Maximization with Thresholds and Fixed Powers}
\label{alg:capacitymaximizationfixedpowers}
\end{algorithm}

\begin{theorem}
Algorithm~\ref{alg:capacitymaximizationfixedpowers} yields a constant-factor approximation.
\end{theorem}

The analysis in \cite{Halldorsson2011} builds on a very involved argument, using a so-called \emph{Red-Blue Lemma}. Essentially the idea is to match links in the computed solution and the optimal one. If the algorithm's solution was much smaller than the optimal one, one link would be left unmatched and thus taken by the algorithm. Our proof in contrast is much simpler and uses the same structure as the one for Theorem~\ref{theorem:unlimitedpowerapproximationfactor}. Again, it is most important to characterize optimal solutions.

\begin{lemma}
\label{lemma:feasiblesetcharacterization}
For each feasible set $\L$ there is a subset $\L' \subseteq \L$ with $\lvert \L' \rvert = \Omega(\lvert \L \rvert)$ and $\sum_{\ell' \in \L', \pi(\ell') < \pi(\ell)} a_p(\ell, \ell') + a_p(\ell', \ell) = O(1)$ for all $\ell \in \R$.
\end{lemma}

Using this lemma, we can adopt the proof that the set $\L$ is at most a constant-factor smaller than the optimal solution literally from the one of Theorem~\ref{theorem:unlimitedpowerapproximationfactor}. Using Markov inequality, one can see that the final solution has size at least $\lvert \L \rvert / 2$. The formal proofs can be found in the appendix.

%% file: limitedpowers.tex
\section{Capacity Maximization with Thresholds and Limited Powers}
\label{sec:limitedpowers}
Having found algorithms for the threshold problem that chooses powers from an unbounded set and for the one with fixed transmission powers, we are now ready to combine these two approaches to an algorithm that chooses transmission powers from a bounded set $[0, p_{\max}]$. The general idea has already been presented by Wan et al.~\cite{Wan2011}.

We decompose the set $\R$ into two sets $\R_1 = \{ \ell \in \R \mid \beta(\ell) \noise d(\ell)^\alpha \leq p_{\max}/4 \}$; $\R_2 = \R \setminus \R_1$. On the set $\R_1$, we run a slightly modified algorithm for unlimited transmission powers. Due to the definition of the set $\R_1$ and the algorithm, it is guaranteed that all assigned transmission powers are at most $p_{\max}$. On the set $\R_2$, we run the fixed-power algorithm setting $p(\ell) = p_{\max}$ for all $\ell \in \R_2$. In the end, we return the better one of the two solutions.

\begin{algorithm}
let $\R_1 = \{ \ell \in \R \mid \beta(\ell) \noise d(\ell)^\alpha \leq p_{\max}/4 \}$; $\R_2 = \R \setminus \R_1$ \;
initialize $\L = \emptyset$ \;
\For{$\ell' \in \R_1$ in decreasing order of $\pi$ values}{
\If{$\sum_{\ell \in \L_1'} w(\ell, \ell') \leq \tau$}{
add $\ell'$ to $\L_1'$\;
}
}
\For{$\ell \in \L_1'$ in increasing order of $\pi$ values}{
\If{$\sum_{\ell' \in \L} w(\ell, \ell') \leq \nicefrac{1}{4}$}{
add $\ell'$ to $\L_1$\;
set $p(\ell) = 2 \beta(\ell) \noise + \beta(\ell) \sum_{\ell' \in \L, \pi(\ell') < \pi(\ell) p(\ell')} p(\ell') d(s, r)^\alpha / d(s', r)^\alpha$\;
}
}
run Algorithm~\ref{alg:capacitymaximizationfixedpowers} on $\R_2$ with $p(\ell) = p_{\max}$ for all $\ell$, let $\L_2$ be the result\;
if $\lvert \L_1 \rvert \geq \L_2$ return $(\L_1, p)$, otherwise return $(\L_2, p_{\max})$\;
\caption{Capacity Maximization with Thresholds and Limited Powers}
\label{alg:capacitymaximizationlimitedpowers}
\end{algorithm}

For this algorithm, we have to show that it computes feasible solutions and that the approximation factor is indeed constant.

\subsection{Feasibility}
In this case, showing feasibility requires two results. On the one hand, we again need that the SINR constraint is fulfilled. On the other hand, we have to show that the assigned transmission powers are no larger than $p_{\max}$.

\begin{theorem}
Let $(\L, p)$ be the solution returned by the algorithm. Then we have $\gamma_\ell(\L, p) \geq \beta(\ell)$ for all $\ell \in \L$.
\end{theorem}

\begin{proof}
For the solution $(\L_2, p_{\max})$ this is trivial due to the definition of Algorithm~\ref{alg:capacitymaximizationfixedpowers}. For $(\L_1, p)$, we can essentially use the proof of Theorem~\ref{theorem:unlimitedpower:feasibility}: Letting the unlimited-power algorithm run on $\L_1$ as the input, it would return exactly $(\L_1, p)$ as the output, which is known to be feasible.
\end{proof}

\begin{theorem}
All transmission powers assigned by Algorithm~\ref{alg:capacitymaximizationlimitedpowers} are at most $p_{\max}$.
\end{theorem}

\begin{proof}
For the solution $(\L_2, p_{\max})$ this is again trivial. For the solution $(\L_1, p)$, we prove the claim by induction. Consider a link $\ell' \in \L_1$. By induction hypothesis, we know that for all $\ell \in \L_1$ with $\pi(\ell) < \pi(\ell')$, we have $p(\ell) \leq p_{\max}$. For power $p(\ell')$ this yields
\begin{align*}
p(\ell') & = 2 \beta(\ell') \noise d(s', r')^\alpha + 2 \beta(\ell) \sum_{\substack{\ell = (s, r) \in \L_1 \\ \pi(\ell) < \pi(\ell')}} \frac{p(\ell)}{d(s, r')^\alpha} d(s', r')^\alpha \\
& \leq 2 \beta(\ell') \noise d(s', r')^\alpha + 2 p_{\max} \sum_{\substack{\ell = (s, r) \in \L_1 \\ \pi(\ell) < \pi(\ell')}} \beta(\ell) \frac{d(s', r')^\alpha}{d(s, r')^\alpha} \\
& \leq 2 \beta(\ell') \noise d(s', r')^\alpha + 2 p_{\max} \sum_{\ell \in \L_1} w(\ell, \ell') \enspace.
\end{align*}
Since we have $4 \beta(\ell') \noise d(s', r') \leq p_{\max}$ and $\sum_{\ell \in \L_1} w(\ell, \ell') \leq \nicefrac{1}{4}$ it follows that $p(\ell') \leq p_{\max}$.
\end{proof}

\subsection{Approximation Factor}

\begin{theorem}
Algorithm~\ref{alg:capacitymaximizationlimitedpowers} computes a constant-factor approximation.
\end{theorem}

\begin{proof}
For $i \in \{1, 2\}$, let $\OPT_i \subseteq \R_i$ be the optimal solution of the problem on the input $\R_i$. Obviously, $\lvert \OPT_1 \rvert + \lvert \OPT_2 \rvert$ is an upper bound on the size of the optimal solution on $\R$. We claim that $\lvert \L_i \rvert = \Omega(\lvert \OPT_i \rvert)$ for $i \in \{1, 2\}$. This then shows the claim.

First, we show $\lvert \L_1 \rvert = \Omega(\lvert \OPT_i \rvert)$. The set $\L_1'$ is exactly the output of the unlimited-power algorithm run on $\R_1$. Thus, by Theorem~\ref{theorem:unlimitedpowerapproximationfactor}, we have $\lvert \L_1' \rvert = \Omega(\lvert \OPT_1 \rvert)$. Furthermore, for all $\ell \in \L_1' \setminus \L_1$, we have $\sum_{\ell' \in \L_1'} w(\ell, \ell') > \nicefrac{1}{4}$. Summing all these inequalities, we get $\sum_{\ell' \in \L_1 \setminus \L_1'} \sum_{\ell' \in \L_1'} w(\ell, \ell') > \nicefrac{1}{4} \lvert \L_1' \setminus \L_1 \rvert$. On the other hand, we have for all $\ell' \in \L_1'$ that $\sum_{\ell \in \L_1' \setminus \L_1} w(\ell, \ell') \leq \tau$. Thus, we have $\sum_{\ell' \in \L_1'} \sum_{\ell \in \L_1' \setminus \L_1} w(\ell, \ell') \leq \tau \cdot \lvert \L_1'\rvert$. Thus, we get $\lvert \L_1' \setminus \L\rvert \leq 4 \tau \lvert \L_1' \rvert$. As $\tau \leq \nicefrac{1}{8}$, this yields $\lvert \L_1 \rvert \geq \nicefrac{1}{2} \cdot \lvert \L_1' \rvert = \Omega(\lvert \OPT_1 \rvert)$.

In order to show $\lvert \L_2 \rvert = \Omega(\lvert \OPT_2 \rvert)$, we first observe that the power assignment making $\OPT_2$ feasible can only use transmission powers from the interval $[p_{\max} / 4, p_{\max}]$. Following the proof of Lemma~\ref{lemma:feasiblesetcharacterization}, this shows that there has to be a subset $\OPT_2' \subseteq \OPT_2$ with $\lvert \OPT_2' \rvert = \Omega(\lvert \OPT_2 \rvert)$ and 
\[
\sum_{\ell' \in \OPT_2', \pi(\ell') < \pi(\ell)} a_{p_{\max}}(\ell, \ell') + a_{p_{\max}}(\ell, \ell') = O(1) \text{ for all $\ell \in \R_2$} \enspace.
\]
Thus, we have $\lvert \L_2 \rvert = \Omega(\lvert \OPT_2 \rvert)$.
\end{proof}

%% file: algorithmwithrates.tex
\section{Capacity Maximization with General Utilities}
In the previous sections we presented constant-factor approximations for the different variants of the capacity maximization problem with individual thresholds. In order to solve the flexible-rate problem with general utility function, we use the respective threshold algorithm as a building block. Inspired by an approach by Halld\'orsson to solve maximum weighted independent set in graphs \cite{Halldorsson2000}, we perform the following steps.

For each link $\ell \in \R$, we determine the maximum utility that it can achieve, referred to as $u_\ell^{\max}$. This value is achieved if only this link is selected and (in case we can select powers) transmits at maximum power. Let the maximum of all maximum link utilities be called $B$. It is a lower bound on the value of the optimal solution. The optimal solution is in turn upper bounded by $n \cdot B$. For all $i \in \{0, 1, 2,\ldots,\lceil\log n\rceil\}$, we run the following procedure. For each link $\ell \in \R$, we query each utility function for the minimum SINR necessary to have utility at least $2^{-i} \cdot B$. This value is taken as the individual threshold of the respective link. On these thresholds, we run the respective algorithm from the previous sections. It returns a set of links and possibly (depending on the problem variant) also a power assignment. As the output, we take the best one among these $\lceil \log n\rceil + 1$ solutions.

\begin{algorithm}
determine $B := \max_{\ell \in \R} u_\ell^{\max}$ \;
\For{$i \in \{0, 1, 2,\ldots,\lceil\log n\rceil\}$}{
set $\beta(\ell) = \min \{ \gamma_\ell \mid u_\ell(\gamma_\ell) \geq 2^{-i} \cdot B \}$ for all $\ell \in \R$ \;
run respective algorithm with these thresholds, let $S_i$ be the solution\;
}
return the best one of the solutions $S_i$\;
\caption{Capacity Maximization with Flexible Rates}
\label{alg:capacitymaximizationwithrates}
\end{algorithm}

\begin{theorem}
\label{theorem:capacitymaximizationwithrates}
Algorithm~\ref{alg:capacitymaximizationwithrates} computes an $O(\log n)$-approximation.
\end{theorem}

\begin{proof}
Consider the set $\L$ of links selected in the optimal solution and let $\gamma_\ell$ be the SINR of link $\ell \in \L$ in this solution. That is, the optimal solution has value $u(\OPT) = \sum_{\ell \in \L} u_{\ell}(\gamma_\ell)$. By definition, we have $u(\OPT) \geq B$.

Furthermore, for $i \in \{0, 1, 2,\ldots,\lceil\log n\rceil\}$, let $\L_i = \{ \ell \in \L \mid u_\ell(\gamma_\ell) \geq 2^{-i} \cdot B \}$, i.e., the set of links having utility at least $2^{-i} \cdot B$ in the optimal solution. Some links might not be contained in any of these sets. Let therefore be $X = \L \setminus \L_{\lceil \log n \rceil + 1}$. We have that $\sum_{\ell \in X} u_\ell(\gamma_\ell) \leq n \cdot 2^{-\lceil \log n \rceil -1} \cdot B \leq b(\OPT) / 2$. So these links hardly contribute to the value of the solution and we leave them out of consideration.

Considering only the sets $\L_i$ defined above, we get $\sum_{i} 2^{-i} \cdot B \cdot \lvert \L_i \rvert \geq u(\OPT) / 4$. Furthermore, solution $S_i$ approximates the set $\L_i$ by a constant factor in size. Expressing this in terms of the utility we have $u(S_i) = \Omega(2^{-i} \cdot B \cdot \lvert \L_i \rvert)$.

Taking the sum over all $i$, we get $\sum_{i} u(S_i) = \Omega(\sum_{i} 2^{-i} \cdot B \cdot \lvert \L_i \rvert) = \Omega(u(\OPT))$ and therefore $\max_{i} u(S_i) = \Omega( u(\OPT) / \log n)$.
\end{proof}

%% file: latencyminimization.tex
\section{Latency Minimization}
Apart from capacity maximization, the other common problem considered in optimization with respect to SINR models is latency minimization. In the standard threshold case this means finding a schedule of shortest possible length in which transmission is carried out successfully at least once. For the case of fixed transmission powers, there are very simple distributed approximation algorithms \cite{Kesselheim2010,Halldorsson2011ICALP}. With results shown above, not only these results but also the ones repeatedly applying a capacity-maximization algorithm can be easily transferred to the setting of individual thresholds.

Studying flexible data rates, this is different. Let us consider the following generalization. As before, we assume that each link has a utility function $u_\ell$ expressing the amount of data that can be transmitted in a time slot depending on the SINR. Furthermore there is a demand $\delta_{\ell}$ representing the amount of data that has to be transmitted via this link. The task is to find a schedule fulfilling these demands, i.e. a sequence of subsets and possibly a power assignment for each step, such that $\sum_t u_\ell(\gamma_\ell^{(t)}) \geq \delta_\ell$. Here, $\gamma_\ell^{(t)}$ denotes the SINR of link $\ell$ in step $t$. 

In order to solve this problem, we consider Algorithm~\ref{alg:latencyminimization}. It repeatedly applies the flexible-rate algorithm on the remaining demands and uses each set for a time slot. The crucial part is which utility functions are considered. The idea is to modify the utility functions by scaling and rounding such that in each round a reasonable step is guaranteed. For our algorithm, we use two differently modified utility functions and we take the shorter one of the two schedules. This also means that the running time of the algorithm can be made linear in the schedule length and polynomial in $n$ by computing both schedules in parallel. For simplicity of the presentation, we neglect this aspect and pretend the two schedules were computed consecutively.  

\begin{algorithm}
for $\ell \in \R$ set $u_{1, \ell}(\gamma_\ell) = \frac{1}{2n} \left\lfloor \frac{2n \cdot u_\ell(\gamma_\ell)}{\delta_\ell} \right\rfloor$, $\delta_{1, \ell} = 1$, $u_{2, \ell}(\gamma_\ell) = \frac{u_\ell(\gamma_\ell)}{u_\ell^{\max}}$, $\delta_{2, \ell} = \frac{\delta_\ell}{u_\ell^{\max}}$ \;
\For{$i \in \{1, 2\}$}{
\While{$\sum_{\ell \in \R} \delta_{i, \ell} > 0$}{
run Algorithm~\ref{alg:capacitymaximizationwithrates} w.r.t. $u'_{i, \ell}(\gamma_\ell) = \min\{\delta_{i, \ell}, u_{i, \ell}(\gamma_\ell) \}$ \;
update $\delta_{i, \ell} := \delta_{i, \ell} - u'_{i, \ell}(\gamma_\ell)$ for all $\ell$, where $\gamma_\ell$ is the achieved SINR\;
}
}
return the shorter schedule
\caption{Latency Minimization}
\label{alg:latencyminimization}
\end{algorithm}

\begin{theorem}
\label{theorem:latencyminimization}
Algorithm~\ref{alg:latencyminimization} computes an $O(\log^2 n)$-approximation.
\end{theorem}

\begin{proof}
Let $T$ be the length of the optimal schedule. Note that scaling utility functions and demands by the same factor does not affect this schedule length. This does not hold for the rounding, which will be discussed below. We distinguish between the two cases $T \geq n$ and $T < n$.

\paragraph{Case 1: $\boldsymbol{T < n}$} In the first case, we only consider the schedule that is computed for $i = 1$, i.e., with respect to the utility functions $u_{1, \cdot}$. Our first observation is that the optimal schedule length with respect to these utility functions (in spite of the rounding) is at most $2 T$. This is due to the following fact: Considering the optimal schedule with respect to the scaled original utility functions, the utility of each link in each step is reduced by at most $\nicefrac{1}{2n}$ in each step. That is, the total utility of a link after rounding still sums up to at least $\nicefrac{1}{2}$. Repeating the same schedule entirely, we get a schedule of length $2T$ in which each link's utility sums up to at least $1$.

Let us consider for each $\ell \in \R$ and $t$ the remaining demand $\delta_{1, \ell}^{(t)}$ of this link after $t$ rounds. We have that all these demands can be fulfilled in a schedule of length at most $2 T$. That is, the maximum utility solution with respect to these demands is at least $\frac{1}{2 T} \sum_{\ell \in \R} \delta_{1, \ell}^{(t)}$. Algorithm~\ref{alg:capacitymaximizationwithrates} finds a $\psi$-approximation of this set for $\psi = O(\log n)$. That is, it yields a solution of utility at least $\frac{1}{\psi 2 T} \sum_{\ell \in \R} \delta_{1, \ell}^{(t)}$. Therefore, we have
\[
\sum_{\ell \in \R} \delta_{1, \ell}^{(t + 1)} \leq \left( 1 - \frac{1}{\psi 2 T} \right) \sum_{\ell \in \R} \delta_{1, \ell}^{(t)} \enspace.
\]
Let us consider some round $t$ with $t > 4 \psi T \cdot \ln(2n)$. We have
\[
\sum_{\ell \in \R} \delta_{1, \ell}^{(t)} \leq \left( 1 - \frac{1}{\psi 2 T} \right)^{t} \sum_{\ell \in \R} \delta_{1, \ell} = \left( 1 - \frac{1}{\psi 2 T} \right)^{t} n < \left( 1 - \frac{1}{\psi 2 T} \right)^{4 \psi T \cdot \ln(2n)} n \leq \frac{1}{2n} \enspace.
\]
By rounding the utility functions, we guaranteed that all $\delta_{1, \ell}^{(t)}$ are multiples of $\nicefrac{1}{2n}$. Thus, showing that their sum is less than $\nicefrac{1}{2n}$, we know that $\delta_{1, \ell}^{(t)} = 0$ for all $\ell$. Therefore, the schedule must be complete by this point.

\paragraph{Case 2: $\boldsymbol{T \geq n}$} In this case, we consider the schedule computed with respect to the utility functions $u_{2, \cdot}$, that is for $i = 2$. Let $\delta^{(t)}_{2, \ell}$ be the remaining demand of link $\ell$ after the $t$th iteration of the while loop. As we scaled the utility function accordingly, the minimal schedule length if only link $\ell$ existed would be $\lceil \delta_{2, \ell} \rceil$. Thus, we have $\sum_{\ell \in \R} \delta_{2, \ell} \leq n \cdot T$.

Now consider the $\delta^{(t)}_{2, \ell}$ values after $t \geq \psi \cdot T \cdot \ln n$ rounds. We have
\[
\sum_{\ell \in \R} \delta^{(t)}_{2, \ell} \leq \left( \frac{1}{n} \right) \sum_{\ell \in \R} \delta_{2, \ell} \leq T \enspace. 
\]
Furthermore, we can observe the following fact: Each singleton set $\{ \ell \}$ either fully satisfies the remaining demand of link $\ell$ or has utility $1$. Thus in each round the algorithm will satisfy the demand of one link or compute a set of summed utility at least $1$. After step $t$ this can in total be at most $n + T \leq 2 T$ steps. Thus, the schedule will have length at most $t + 2 T = O(T \cdot \log^2 n)$. 
\end{proof}

%% file: thresholdssmallerthanone.tex
\section{Thresholds Smaller than 1}
\label{sec:thresholdssmallerthanone}
When analyzing the algorithms in this paper, we assumed that SINR values below $1$ have zero utility. That also means that all thresholds are at least $1$. Although not stated explicitly, this restriction also applies to related work \cite{HalldorssonINFOCOM2012,Santi2009}. Due to the signal-strengthening technique this can be generalized to settings in which the thresholds are at least $\beta_{\min}$. However, if this lower bound is not considered as a constant, this worsens the approximation factors. In this section, we show that this is not a specific issue of our analysis but a more general problem of the two entire classes of simple algorithms that have been considered so far. In particular, we show that approximation factors will be $\Omega(1 / \beta_{\min})$ for all greedy and ALOHA-like algorithms.

The algorithms for capacity maximization in this paper but also in most other works \cite{Halldorsson2009,Halldorsson2011,Kesselheim2011} follow the following structure. The main part is a greedy selection of the link set. That is, the algorithm first orders the links by some local property such as the length. In the second step it iterates over the ordered links and adds a link to the set after checking some property. This link stays in the set respectively cannot be added to the set anymore. In the last step, one might have to remove links whose constraints are not fulfilled by link added later on \cite{Halldorsson2011}.

\begin{theorem}
The approximation factor of every deterministic greedy algorithm is in $\Omega(1 / \beta_{\min})$.
\end{theorem}

\begin{proof}
We consider instances of the following kind. We consider the line, identified with $\RR$. For each link, either the sender is located at $0$ and the receiver at $1$ (sender left of receiver) or the sender is located at $1$ and the receiver at $0$ (sender right of receiver). The threshold for each link is $\beta_{\min}$. In terms of the processing order, none of these links are distinguishable for the algorithm since each one has the same local properties. For this reason, we can consider the algorithm like an online algorithm. 

The algorithm has to accept the first link in the processing order because otherwise the approximation factor would be unbounded if this was the only link. Our instance, however, is designed such that $k = \lfloor 1 / \beta_{\min} \rfloor$ ``reversed'' links follow. All of these links can be accepted in the optimal solution, whereas none can be accepted by the algorithm. Even adding a clean-up step like in \cite{Halldorsson2011} does not change the situation as this clean-up step only removes infeasible links. In our instance, however, if one link is infeasible all are.

Thus, the approximation factor is $k / 1 = \Omega(1 / \beta_{\min})$.
\end{proof}

The second common approach \cite{Fanghaenel2009,Kesselheim2010,Halldorsson2011ICALP} is to minimize latency by a randomized algorithm. Here, each sender transmits in each step with a certain probability until the transmission is successful. We consider protocols in which links may use different transmission probabilities over time but all links use the same probability in a time slot. For these algorithms also an $\Omega(1 / \beta_{\min})$ holds. 

\begin{theorem}
Every ALOHA-like algorithm for latency minimization has an approximation factor in $\Omega(1 / \beta_{\min})$.
\end{theorem}

\begin{proof}
Let $k = \lfloor 1 / \beta_{\min} \rfloor$. Consider the following instance on the line. There are $k$ links having the sender at $0$ and the receiver at $1$, and there are $k$ further links whose sender is at $1$ and whose receiver is at $0$. The optimal schedule length in this case is $2$. We show that any symmetric randomized algorithms yields a schedule of length $\Omega(k)$ with constant probability.

We assume binary feedback, that is each sender gets to know after a transmission attempt whether its transmission has been successfully received (in this case it drops out). Furthermore, we assume all senders follow the same algorithm. Consequently, running the algorithms can be described by the sequence of transmission probabilities $p_t$, which all remaining senders use in round $t$.

We now consider the time $T$ until at least $k/2$ transmissions have been carried out. We define $n_t$ to be the remaining number of transmissions after round $t$, $n_0 := 2 k$. Furthermore, we define the random variable $\delta_t$ as $n_t - n_{t+1}$ if $n_t \geq \frac{3}{2} k$ and $0$ otherwise. We have $\Ex{\delta_t} \leq 4$ because of the following reason. In the case $n_t < \frac{3}{2} n$ this is clear by definition. In the case $n_t \geq \frac{3}{2} n$, in contrast, it is the expected number of successful transmissions. These transmissions are exposed to the interference of at least $k / 2$ reverse links. Therefore, we have
\[
\Ex{\delta_t \mid n_t \geq \frac{3}{2} k} \leq 2 k p_t (1-p_t)^{k/2} \leq 2 k \frac{2}{k+2} \left( 1- \frac{2}{k+2} \right)^{k/2} \leq 2 \frac{2k}{k + 2} \leq 4 \enspace.
\]
Furthermore, observe that $T \leq \frac{k}{16}$ implies $2k - \sum_{t=1}^{k/16} \delta_t \leq \frac{3}{2} k$ since $T$ was defined such that $n_t \geq \frac{3}{2} n$ for all $t \leq T$. Thus, we can get a bound on $T$ as follows
\begin{align*}
\Pr{T \leq \frac{k}{16}} & \leq \Pr{2k - \sum_{t=1}^{k/16} \delta_t \leq \frac{3}{2} k} = \Pr{\sum_{t=1}^{k/16} \delta_t \geq \frac{1}{2} k} \\
& \leq \frac{2}{k} \cdot \Ex{\sum_{t=1}^{k/16} \delta_t} \leq \frac{2}{k} \cdot 4 \cdot \frac{k}{16} \leq \frac{1}{2} \enspace.
\end{align*}
That is, with constant probability we need at least $\frac{k}{16} = \Omega(1 / \beta_{\min})$ time slots.
\end{proof}

%% file: discussion.tex
\section{Conclusion}
In this paper, we presented the flexible-rate capacity maximization problem and a first approximation algorithm. Maybe surprisingly, understanding the threshold variants brings about a reasonable approximation factor for this case as well.

The most striking limitation is that these approaches are not able to deal with SINR values smaller than $1$ respectively a constant. Indeed, by the use of forward error-correcting codes, this case can actually occur in practice. However, we could find evidence that one will have to come up with new techniques to solve this problem.

The other major open issue is to improve the approximation factor. As a first step, one could ask whether there is a constant-factor approximation for \emph{weighted} capacity maximization with thresholds. That is, each link has a threshold (or there is a global one) and a weight. The task is to select a feasible set of maximum weight. For linear power assignments, a constant-factor approximation is known \cite{HalldorssonINFOCOM2012} but all other approaches are only able to give $O(\log n)$ as well.

Another interesting topic for future research could include different objective functions. So far, we only maximized the sum of all link utilities. While it is trivial to maximize the minimal utility in a centralized way, it still remains an open problem to define further fairness aspects and solve the problem accordingly.

%% file: omittedproofs.tex
\input{preliminaries}

\section{Proof of Lemma~\ref{lemma:admissiblesetcharacterization}}
Lemma~\ref{lemma:admissiblesetcharacterization} states a bound for all $\ell \in \R$, not only the ones that belong to $\L$ or $\L'$. This is due to the fact that the links are located in a metric space and for a node $\ell$ that does not belong to $\L$, we can use a nearby node as a witness. In the following two lemmas, we present the general technique, that will be useful later on as well. This technique has essentially been introduced in \cite{Fanghaenel2009}.

\begin{lemma}
\label{lemma:admissibleset:geometrysenders}
Let $w$ be a weight function on links having the following property: Given three links $\ell = (s,r), \ell'=(s',r'), \ell''=(s'',r'') \in \R$ such that $d(s, r') \leq 3 d(s, r'')$ and $\pi(\ell') < \pi(\ell'')$ then $w(\ell, \ell'') \leq c \cdot w(\ell, \ell')$ for some constant $c$. Furthermore, let $\L \subseteq \R$ be a set such that $\L$ or $\L^\ast$ is admissible and such that $\sum_{\ell' \in \L} w(\ell', \ell) = O(1)$ for all $\ell \in \L$. Then we have $\sum_{\ell' \in \L, \pi(\ell') < \pi(\ell)} w(\ell', \ell) = O(1)$ for all $\ell \in \R$.
\end{lemma}

\begin{figure}
\begin{center}
\begin{tikzpicture}
        \tikzstyle{vertex}=
        [%
          fill=white,%
          draw=black,%
          minimum size=1mm,%
          circle,%
          thick%
        ]
        
\coordinate (s) at (-2, 0);
\coordinate (r) at (0,0);

\coordinate (sk) at (5,1);
\coordinate (rk) at (3,0);

\coordinate (s1) at (-0.2,1);
\coordinate (r1) at (-1.8,3);

\coordinate (s2) at (0.8,-0.8);
\coordinate (r2) at (2.5,-2);

\coordinate (s3) at (-2,-1);
\coordinate (r3) at (-3.5,-2);

\coordinate (s4) at (1.8,3.2);
\coordinate (r4) at (4.4,2.4);

        \draw[fill=black,opacity=.2] (r) circle (3);

        \draw[fill=black,opacity=.2] (r) circle (1.5);

        \foreach \x in {1,2,...,4}
        {
            \node (ns\x) [vertex, label=right:$s_\x$] at (s\x) {};
            \node (nr\x) [vertex, label=right:$r_\x$] at (r\x) {};
        }
 
        \foreach \x in {1,2,...,4} \path [thick,shorten >=1pt,-stealth'] (ns\x) edge (nr\x);

        \node (ns) [vertex,fill=black] at (s) {};
        \node (nr) [vertex,fill=black] at (r) {};
        
        \path [thick,shorten >=1pt,-stealth'] (ns) edge (nr);
        
        \node (nsk) [vertex, label=right:$s_k$] at (sk) {};
        \node (nrk) [vertex, label=right:$r_k$] at (rk) {};
        
        \path [thick,shorten >=1pt,-stealth'] (nsk) edge (nrk);
        
        \node[text width=5cm,anchor=north west] at (6, 2) {By triangle inequality:
\begin{itemize}
  \item $d(s_i, r_i) \geq \delta$\\ for all $i \in U$,
  \item $d(s_i, s_k) \leq 2 \delta$\\ for all $i, k \in U$, and
  \item $d(s_i, r_j) \leq 3 d(s_i, r)$.
	\end{itemize}};
       
      \end{tikzpicture}
      \end{center}
\caption{Visualization of the geometry of the link $\ell$ (black nodes) and the set $\L$ (white nodes) in the proof of Lemma~\ref{lemma:admissibleset:geometrysenders}. There is no receiver of $\L$ inside the light shaded area. The links in $\L$ whose senders lie within the dark shaded area (in this example $1$ and $2$), belong to $U$.}
\end{figure}

\begin{proof}
Fix a link $\ell = (s, r) \in \R$. We consider the link set $\{ \ell' \in \L \mid \pi(\ell') < \pi(\ell) \}$. Let the links in this set be denoted by $\ell_1, \ldots, \ell_{\bar{n}}$. Let furthermore be $\ell_i = (s_i, r_i)$  for all $i \in [\bar{n}]$.

Let $k$ be the index of the receiver $r_1$,\ldots,$r_{\bar{n}}$ that is closest to $r$, that is $k \in \arg\min_{i \in [\bar{n}]} d(r_i, r)$. We define the set $U$ to be the indices of links whose senders $s_i$ lie within a distance of at most $\delta := \frac{1}{2} d(r_k, r)$ from $r$, i.\,e. $U = \{ i \in [\bar{n}] \mid d(s_i, r) \leq \delta \}$.

By triangle inequality, we have for all $i \in U$ that $2 \delta \leq d(r, r_i) \leq d(s_i, r_i) + d(s_i, r) \leq d(s_i, r_i) + \delta$. Thus $d(s_i, r_i) \geq \delta$.

We now bound the size of $U$. If $U = \emptyset$ this is trivial. Otherwise, fix an arbitrary $j \in U$. For all $i \in U$, we have $d(s_i, s_j) \leq d(s_i, r) + d(s_j, r) \leq 2 \delta$. That is, by triangle inequality
\begin{align*}
d(s_i, r_j) \leq d(s_i, s_j) + d(s_j, r_j) \leq 2 \delta + d(s_j, r_j) \leq 3 d(s_j, r_j) \\
\text{and } \qquad d(s_j, r_i) \leq d(s_i, s_j) + d(s_i, r_i) \leq 2 \delta + d(s_i, r_i) \leq 3 d(s_i, r_i) \enspace.
\end{align*}
These inequalities yield
\[
\sum_{i \in U \setminus \{j\}} \beta(\ell_i) \beta(\ell_j) \frac{d(s_i, r_i)^\alpha d(s_j, r_j)^\alpha}{d(s_i, r_j)^\alpha d(s_j, r_i)^\alpha} \geq \sum_{i \in U \setminus \{j\}} \frac{d(s_i, r_i)^\alpha d(s_j, r_j)^\alpha}{9^\alpha d(s_i, r_i)^\alpha d(s_j, r_j)^\alpha} \geq \frac{\lvert U \rvert - 1}{9^\alpha} \enspace.
\]
Furthermore, as $\L$ or $\L^\ast$ is admissible, we have
\[
\sum_{i \in U \setminus \{j\}} \beta(\ell_i) \beta(\ell_j) \frac{d(s_i, r_i)^\alpha d(s_j, r_j)^\alpha}{d(s_i, r_j)^\alpha d(s_j, r_i)^\alpha} \leq 1 \enspace.
\]
Thus, $\lvert U \rvert \leq 9^\alpha + 1$.

Let us now proceed to the terms that do not belong to $U$. For all $i \in [\bar{n}] \setminus U$ it holds that
\begin{align*}
d(s_i, r_k) & \leq d(s_i, r) + d(r_k, r) && \text{by triangle inequality}\\
& \leq d(s_i, r) + 2d(s_i, r) && \text{by definition of $U$}\\
& = 3d(s_i, r)\enspace.
\end{align*}
As we have $\pi(\ell_k) < \pi(\ell)$, we have $w(\ell_i, \ell) \leq c \cdot w(\ell_i, \ell_k)$ for some constant $c$.

Thus in total, we get
\[
\sum_{\ell' \in \L'} w(\ell', \ell) \leq \lvert U \rvert + \sum_{i \not\in U} w(\ell_i, \ell) \leq \lvert U \rvert + c \sum_{i \not\in U} w(\ell_i, \ell_k) = O(1) \enspace.
\]
\end{proof}

We can also swap roles of senders and receivers in the lemma.

\begin{lemma}
\label{lemma:admissibleset:geometryreceivers}
Let $w$ be a weight function on links having the following property: Given three links $\ell = (s,r), \ell'=(s',r'), \ell''=(s'',r'') \in \R$ such that $d(s', r) \leq 3 d(s'', r)$ and $\pi(\ell') < \pi(\ell'')$ then $w(\ell, \ell'') \leq c \cdot w(\ell, \ell')$ for some constant $c$. Furthermore, let $\L \subseteq \R$ be a set such that $\L$ or $\L^\ast$ is admissible and such that $\sum_{\ell' \in \L} w(\ell', \ell) = O(1)$ for all $\ell \in \L$. Then we have $\sum_{\ell' \in \L, \pi(\ell') < \pi(\ell)} w(\ell', \ell) = O(1)$ for all $\ell \in \R$.
\end{lemma}

This lemma can be shown by exchanging senders and receivers and vice versa in the proof of Lemma~\ref{lemma:admissibleset:geometrysenders} without further adaptations.

\begin{proof}
Fix a link $\ell = (s, r) \in \R$. We consider the link set $\{ \ell' \in \L \mid \pi(\ell') < \pi(\ell) \}$. Let the links in this set be denoted by $\ell_1, \ldots, \ell_{\bar{n}}$. Let furthermore be $\ell_i = (s_i, r_i)$  for all $i \in [\bar{n}]$.

Let $k$ be the index of the sender $s_1$,\ldots,$s_{\bar{n}}$ that is closest to $s$, that is $k \in \arg\min_{i \in [\bar{n}]} d(s_i, s)$. We define the set $U$ to be the indices of links whose senders $r_i$ lie within a distance of at most $\delta := \frac{1}{2} d(s_k, s)$ from $s$, i.\,e. $U = \{ i \in [\bar{n}] \mid d(r_i, s) \leq \delta \}$.

By triangle inequality, we have for all $i \in U$ that $2 \delta \leq d(s, s_i) \leq d(r_i, s_i) + d(r_i, s) \leq d(r_i, s_i) + \delta$. Thus $d(r_i, s_i) \geq \delta$.

To bound the size of $U$ (unless it is empty), fix an arbitrary $j \in U$. For all $i \in U$, we have $d(r_i, r_j) \leq d(r_i, s) + d(r_j, s) \leq 2 \delta$. That is, by triangle inequality
\begin{align*}
d(r_i, s_j) \leq d(r_i, r_j) + d(r_j, s_j) \leq 2 \delta + d(r_j, s_j) \leq 3 d(r_j, s_j) \\
\text{and } \qquad d(r_j, s_i) \leq d(r_i, r_j) + d(r_i, s_i) \leq 2 \delta + d(r_i, s_i) \leq 3 d(r_i, s_i) \enspace.
\end{align*}
These inequalities and the fact that $\L$ or $\L^\ast$ is admissible, yield
\[
\frac{\lvert U \rvert - 1}{9^\alpha} \leq \sum_{i \in U \setminus \{j\}} \beta(\ell_i) \beta(\ell_j) \frac{d(r_i, s_i)^\alpha d(r_j, s_j)^\alpha}{d(r_i, s_j)^\alpha d(r_j, s_i)^\alpha} \leq 1 \enspace.
\]
Thus, $\lvert U \rvert \leq 9^\alpha + 1$.

For all $i \in [\bar{n}] \setminus U$ it holds that $d(r_i, s_k) \leq d(r_i, s) + d(s_k, s) \leq d(r_i, s) + 2d(r_i, s) = 3d(r_i, s)$. As we have $\pi(\ell_k) < \pi(\ell)$, we have $w(\ell_i, \ell) \leq c \cdot w(\ell_i, \ell_k)$ for some constant $c$.

Thus in total, we get
\[
\sum_{\ell' \in \L'} w(\ell', \ell) \leq \lvert U \rvert + \sum_{i \not\in U} w(\ell_i, \ell) \leq \lvert U \rvert + c \sum_{i \not\in U} w(\ell_i, \ell_k) = O(1) \enspace.
\]
\end{proof}

\section{Proof of Lemma~\ref{lemma:admissiblesetcharacterization}}
The actual proof of Lemma~\ref{lemma:admissiblesetcharacterization} now proceeds in four steps, namely Claims~\ref{claim:admissiblesets:doublesum} to \ref{claim:admissiblesets:newweights}. In the end, the results follow from Claims~\ref{claim:admissiblesets:forallinr} and \ref{claim:admissiblesets:newweights}.

\begin{claim}
\label{claim:admissiblesets:doublesum}
For each admissible set $\L$ there is a subset $\L' \subseteq \L$ with $\lvert \L' \rvert = \Omega(\lvert \L \rvert)$ and
\[
\frac{1}{\lvert \L' \rvert} \sum_{\ell \in \L'} \sum_{\substack{\ell' \in \L' \\ \pi(\ell') > \pi(\ell)}} \frac{\beta(\ell') d(s', r')^\alpha}{d(s', r)^\alpha} = O(1) \enspace.
\]
\end{claim}

\begin{proof}
We use Corollary~\ref{corollary:scalethresholdssubset} to consider an admissible set $\L'$ with $\lvert \L' \rvert = \Omega( \lvert \L \rvert)$ that originates from scaling all thresholds to $\beta'(\ell) = 4^\alpha \cdot \beta(\ell)$. As $\L'$ is admissible there is a power assignment $p$ making all SINR conditions fulfilled. Let $\L' = \{ \ell_1, \ell_2, \ldots \}$ with $p(\ell_1) \leq p(\ell_2) \leq \ldots$ and furthermore $\ell_i = (s_i, r_i)$ and $\beta_i = \beta(\ell_i)$ for all $i$. 

We have for all $i$
\[
4^\alpha \beta_i \sum_{j > i} \frac{p(\ell_j)}{d(s_j, r_i)^\alpha} \leq 4^\alpha \beta_i \sum_{j \neq i} \frac{p(\ell_j)}{d(s_j, r_i)^\alpha} \leq \frac{p(\ell_i)}{d(s_i, r_i)^\alpha} \enspace.
\]
Using the fact that $p(\ell_j) \geq p(\ell_i)$ for all $j > i$ this yields
\[
\sum_{j > i} \frac{\beta_i d(s_i, r_i)^\alpha}{d(s_j, r_i)^\alpha} \leq \frac{1}{4^\alpha} \enspace. 
\]

Now let us consider some $i$ and $j$ such that $j > i$ and $\pi(j) > \pi(i)$. Having $j > i$ yields $\frac{\beta_i d(s_i, r_i)^\alpha}{d(s_j, r_i)^\alpha} \leq \frac{1}{4^\alpha}$, whereas $\pi(j) > \pi(i)$ yields $\beta_j d(s_j, r_j)^\alpha \leq \beta_i d(s_i, r_i)^\alpha$. That is, we have by triangle inequality
\begin{align*}
d(s_j, r_i) & \leq d(s_j, r_j) + d(s_i, r_j) + d(s_i, r_i) \\
& \leq 2 \beta_i^{\frac{1}{\alpha}} d(s_i, r_i) + d(s_i, r_j) \\
& \leq \frac{1}{2} d(s_j, r_i) + d(s_i, r_j) \enspace.
\end{align*}
That is, we have $d(s_j, r_i)^\alpha \leq 2^\alpha d(s_i, r_j)^\alpha$.

Thus, we get
\[
\sum_{\substack{j > i \\ \pi(j) < \pi(i)}} \frac{\beta_i d(s_i, r_i)^\alpha}{d(s_j, r_i)^\alpha} + \sum_{\substack{j > i \\ \pi(j) > \pi(i)}} \frac{\beta_j d(s_j, r_j)^\alpha}{\frac{1}{2^\alpha} d(s_i, r_j)^\alpha} \leq \frac{1}{4^\alpha} \enspace.
\]
Summing over all $i$, we get
\[
\sum_i \sum_{\substack{j > i \\ \pi(j) < \pi(i)}} \frac{\beta_i d(s_i, r_i)^\alpha}{d(s_j, r_i)^\alpha} + \sum_i \sum_{\substack{j > i \\ \pi(j) > \pi(i)}} \frac{\beta_j d(s_j, r_j)^\alpha}{\frac{1}{2^\alpha} d(s_i, r_j)^\alpha} \leq \frac{1}{4^\alpha} \enspace.
\]
Now let us consider the first part of the sum. Swapping $i$ and $j$ and reordering the sum afterwards, it is equal to
\[
\sum_i \sum_{\substack{j > i \\ \pi(j) < \pi(i)}} \frac{\beta_i d(s_i, r_i)^\alpha}{d(s_j, r_i)^\alpha} = \sum_j \sum_{\substack{i > j \\ \pi(i) < \pi(j)}} \frac{\beta_j d(s_j, r_j)^\alpha}{d(s_i, r_j)^\alpha} = \sum_i \sum_{\substack{j < i \\ \pi(j) > \pi(i)}} \frac{\beta_j d(s_j, r_j)^\alpha}{d(s_i, r_j)^\alpha} \enspace.
\]
Replacing the sum in the previously obtained bound, we get
\[
\sum_i \sum_{\substack{j < i \\ \pi(j) > \pi(i)}} \frac{\beta_j d(s_j, r_j)^\alpha}{d(s_i, r_j)^\alpha} + \sum_i \sum_{\substack{j > i \\ \pi(j) > \pi(i)}} \frac{\beta_j d(s_j, r_j)^\alpha}{d(s_i, r_j)^\alpha} \leq \frac{1}{2^\alpha} \enspace.
\]
This shows the claim.
\end{proof}

\begin{claim}
\label{claim:admissiblesets:forallinl}
For each admissible set $\L$ there is a subset $\L' \subseteq \L$ with $\lvert \L' \rvert = \Omega(\lvert \L \rvert)$ and for all $\ell = (s, r) \in \L'$ 
\[
\sum_{\ell' = (s', r') \in \L'} \frac{\min\{ \beta(\ell) d(s, r)^\alpha, \beta(\ell') d(s', r')^\alpha \} }{\min\{d(s', r)^\alpha, d(s, r')^\alpha\}} = O(1) \enspace.
\]
\end{claim}

\begin{proof}
Using Lemma~\ref{lemma:dualadmissible}, we know there is a subset $\L''$ of $\L$ with $\lvert \L'' \rvert = \Omega(\lvert \L \rvert)$ such that both $\L''$ and $(\L'')^\ast$ are admissible. Using Claim~\ref{claim:admissiblesets:doublesum}, we know there is a subset $\L'''$ being a constant fraction of $\L''$ such that
\[
\frac{1}{\lvert \L''' \rvert} \sum_{\ell = (s, r) \in \L'''} \sum_{\substack{\ell' = (s', r') \in \L''' \\ \pi(\ell') > \pi(\ell)}} \frac{\beta(\ell') d(s', r')^\alpha}{d(s', r)^\alpha} \leq c
\]
and
\[
\frac{1}{\lvert \L''' \rvert} \sum_{\ell = (s, r) \in \L'''} \sum_{\substack{\ell' = (s', r') \in \L''' \\ \pi(\ell') > \pi(\ell)}} \frac{\beta(\ell') d(s', r')^\alpha}{d(s, r')^\alpha} \leq c
\]
for some suitable constant $c$. Adding up both bounds, we get
\[
\frac{1}{\lvert \L''' \rvert} \sum_{\ell = (s, r) \in \L'''} \sum_{\substack{\ell' = (s', r') \in \L''' \\ \pi(\ell') > \pi(\ell)}} \frac{\beta(\ell') d(s', r')^\alpha}{\min \{ d(s', r)^\alpha, d(s, r')^\alpha \}} \leq 2 c \enspace.
\]
By reordering the sum, we get
\[
\frac{1}{\lvert \L''' \rvert} \sum_{\ell = (s, r) \in \L'''} \sum_{\substack{\ell' = (s', r') \in \L''' \\ \pi(\ell') < \pi(\ell)}} \frac{\beta(\ell) d(s, r)^\alpha}{\min \{ d(s', r)^\alpha, d(s, r')^\alpha \}} \leq 2 c \enspace.
\]
Adding up these two inequalities, we get
\[
\frac{1}{\lvert \L''' \rvert} \sum_{\ell = (s, r) \in \L'''} \sum_{\substack{\ell' = (s', r') \in \L''' \\ \ell' \neq \ell}} \frac{\min\{ \beta(\ell) d(s, r)^\alpha, \beta(\ell') d(s', r')^\alpha \}}{\min \{ d(s', r)^\alpha, d(s, r')^\alpha \}} \leq 4 c \enspace.
\]
By Markov's inequality for half of the links $\ell = (s, r) \in \L'''$ we have
\[
\sum_{\substack{\ell' = (s', r') \in \L''' \\ \ell' \neq \ell}} \frac{\min\{ \beta(\ell) d(s, r)^\alpha, \beta(\ell') d(s', r')^\alpha \} }{\min\{d(s', r)^\alpha, d(s, r')^\alpha\}} \leq 8 c \enspace.
\]
Choose these links as the set $\L'$.
\end{proof}

\begin{claim}
\label{claim:admissiblesets:forallinr}
For each admissible set $\L$ there is a subset $\L' \subseteq \L$ with $\lvert \L' \rvert = \Omega(\lvert \L \rvert)$ and for all $\ell = (s, r) \in \R$
\[
\sum_{\substack{\ell' = (s', r') \in \L' \\ \pi(\ell') < \pi(\ell)}} \min \left\{ 1, \frac{\beta(\ell) d(s, r)^\alpha}{d(s', r)^\alpha} \right\} + \min \left\{ 1, \frac{\beta(\ell) d(s, r)^\alpha}{d(s, r')^\alpha} \right\} = O(1) \enspace.
\]
\end{claim}

\begin{proof}
We show the claim for the set $\L'$ described in Claim~\ref{claim:admissiblesets:forallinl}. For $\ell = (s,r), \ell'=(s',r') \in \R$ define the following weights
\begin{align*}
w_1(\ell, \ell') & = \min \left\{ 1, \frac{\beta(\ell) d(s, r)^\alpha}{d(s, r')^\alpha}, \frac{\beta(\ell') d(s', r')^\alpha}{d(s, r')^\alpha} \right\} \qquad \text{ and} \\
w_2(\ell, \ell') & = \min \left\{ 1, \frac{\beta(\ell) d(s, r)^\alpha}{d(s', r)^\alpha}, \frac{\beta(\ell') d(s', r')^\alpha}{d(s', r)^\alpha} \right\} \enspace.
\end{align*}
By definition of the set $\L'$, we have for both $i=1$ and $i=2$ that $\sum_{\ell \in \L'} w_i(\ell, \ell') = O(1)$ for all $\ell' \in \L'$. We will now apply Lemma~\ref{lemma:admissibleset:geometrysenders} on the weights $w_1$ and Lemma~\ref{lemma:admissibleset:geometryreceivers} on the weights $w_2$. In combination, this will yield the claim.

In order to apply Lemma~\ref{lemma:admissibleset:geometrysenders}, we consider three links $\ell = (s,r), \ell'=(s',r'), \ell''=(s'',r'') \in \R$ such that $d(s, r') \leq 3 d(s, r'')$ and $\pi(\ell') < \pi(\ell'')$. We have
\begin{align*}
w_1(\ell, \ell'') & = \min \left\{ 1, \frac{\beta(\ell) d(s, r)^\alpha}{d(s, r'')^\alpha}, \frac{\beta(\ell'') d(s'', r'')^\alpha}{d(s, r'')^\alpha} \right\} \\
& \leq \min \left\{ 1, \frac{\beta(\ell) d(s, r)^\alpha}{\frac{1}{3^\alpha}  d(s, r')^\alpha}, \frac{\beta(\ell') d(s', r')^\alpha}{\frac{1}{3^\alpha} d(s, r')^\alpha} \right\} \\
& \leq 3^\alpha \min \left\{ 1, \frac{\beta(\ell) d(s, r)^\alpha}{d(s, r')^\alpha}, \frac{\beta(\ell') d(s', r')^\alpha}{d(s, r')^\alpha} \right\} \\
& = 3^\alpha w_1(\ell, \ell') \enspace.
\end{align*}
Thus, by Lemma~\ref{lemma:admissibleset:geometrysenders}, $\sum_{\ell \in \L', \pi(\ell) < \pi(\ell')} w_1(\ell, \ell') = O(1)$ for all $\ell \in \R$. 

Next, we apply Lemma~\ref{lemma:admissibleset:geometryreceivers} on the weights $w_2$. For three links $\ell = (s,r), \ell'=(s',r'), \ell''=(s'',r'') \in \R$ such that $d(s, r') \leq 3 d(s, r'')$ and $\pi(\ell') < \pi(\ell'')$, we get
\begin{align*}
w_2(\ell, \ell'') & = \min \left\{ 1, \frac{\beta(\ell) d(s, r)^\alpha}{d(s'', r)^\alpha}, \frac{\beta(\ell'') d(s'', r'')^\alpha}{d(s'', r)^\alpha} \right\} \\
& \leq \min \left\{ 1, \frac{\beta(\ell) d(s, r)^\alpha}{\frac{1}{3^\alpha}  d(s', r)^\alpha}, \frac{\beta(\ell') d(s', r')^\alpha}{\frac{1}{3^\alpha} d(s', r)^\alpha} \right\} \\
& \leq 3^\alpha \min \left\{ 1, \frac{\beta(\ell) d(s, r)^\alpha}{d(s', r)^\alpha}, \frac{\beta(\ell') d(s', r')^\alpha}{d(s', r)^\alpha} \right\} \\
& = 3^\alpha w_2(\ell, \ell') \enspace.
\end{align*}
Thus, by Lemma~\ref{lemma:admissibleset:geometryreceivers}, $\sum_{\ell \in \L', \pi(\ell) < \pi(\ell')} w_2(\ell, \ell') = O(1)$ for all $\ell \in \R$.
\end{proof}

\begin{claim}
\label{claim:admissiblesets:newweights}
For each admissible set $\L$ there is a subset $\L' \subseteq \L$ with $\lvert \L' \rvert = \Omega(\lvert \L \rvert)$ and
\[
\sum_{\substack{\ell' = (s', r') \in \L' \\ \pi(\ell') < \pi(\ell)}} \min \left\{ 1, \beta(\ell) \beta(\ell') \frac{d(s, r)^\alpha \cdot d(s', r')^\alpha}{d(s, r')^\alpha \cdot d(s', r)^\alpha} \right\} = O(1) \text{ for all $\ell = (s, r) \in \R$} .
\]
\end{claim}

\begin{proof}
We consider the subset described in Claim~\ref{claim:admissiblesets:forallinr}. Furthermore, without loss of generality, by Corollary~\ref{corollary:scalethresholdssubset}, we can assume that all thresholds are scaled to $\beta'(\ell) = \sqrt{2^\alpha(2 + 6^\alpha)} \cdot \beta(\ell)$. Let this be the set $\L'$. 

Now consider some fixed $\ell \in \R$. In order to show the bound, we distinguish between three kinds of links in the sum. We set
\begin{itemize}
  \item $\L_1 = \{ \ell' = (s', r') \in \L' \mid \pi(\ell') < \pi(\ell), d(s', r)^\alpha \geq \beta(\ell') \cdot d(s', r')^\alpha \}$,
  \item $\L_2 = \{ \ell' = (s', r') \in \L' \mid \pi(\ell') < \pi(\ell), d(s, r')^\alpha \geq \beta(\ell') \cdot d(s', r')^\alpha \}$, and
  \item $\L_3 = \{ \ell' = (s', r') \in \L' \mid \pi(\ell') < \pi(\ell), \ell \not\in \L_1 \cup \L_2 \}$.
\end{itemize}
Observe that the bound for the link sets $\L_1$ and $\L_2$ immediately follows from Claim~\ref{claim:admissiblesets:forallinr}. Therefore, in the remainder, we only consider the set $\L_3$.

Let the links in $\L_3$ be referred to as $\ell_1 = (s_1, r_1)$, \ldots, $\ell_{\bar{n}} = (s_{\bar{n}}, r_{\bar{n}})$ with $\beta_i = \beta(\ell_i)$ and $\beta_1 d(s_1, r_1)^\alpha \leq \beta_2 d(s_2, r_2)^\alpha \leq \ldots \leq \beta_{\bar{n}} d(s_{\bar{n}}, r_{\bar{n}})^\alpha$.

Consider some $i \in [\bar{n} - 1]$. We will show the following two bounds:
\begin{enumerate}
  \item $\beta_i \beta_{i+1} d(s_i, r_i)^\alpha d(s_{i + 1}, r_{i + 1})^\alpha \leq \frac{1}{2} d(s, r_{i+1})^\alpha d(s_{i+1}, r)^\alpha$, and
  \item $\beta_{i+1} d(s_{i+1}, r_{i+1})^\alpha \geq 2^i \cdot z$.
\end{enumerate}
Multiplying the (scaled) SINR constraints yields
\begin{equation}
\label{eq:admissiblesets:newweights:sinr}
\frac{\beta_i \beta_{i+1} d(s_i, r_i)^\alpha d(s_{i+1}, r_{i+1})^\alpha}{d(s_i, r_{i+1})^\alpha d(s_i, r_{i+1})^\alpha} \leq \frac{1}{2^\alpha (2 + 6^\alpha)} \enspace.
\end{equation}

By triangle inequality and the definition of $\L'$, we have
\begin{align*}
d(s_i, r_{i + 1}) \leq d(s_i, r) + d(s, r) + d(s, r_{i + 1}) \leq d(s, r_{i+1}) + 2 \beta_i^{1 / \alpha} d(s_i, r_i) \enspace, \quad \text{ and}\\
d(s_{i + 1}, r_i) \leq d(s_{i + 1}, r) + d(s, r) + d(s, r_i) \leq d(s_{i+1}, r) + 2 \beta_i^{1 / \alpha} d(s_i, r_i) \enspace.
\end{align*}
Multiplying the two inequalities yields
\[
d(s_i, r_{i + 1}) d(s_{i + 1}, r_i) \leq \left( d(s, r_{i+1}) + 2 \beta_i^{1 / \alpha} d(s_i, r_i) \right) \left( d(s_{i+1}, r) + 2 \beta_i^{1 / \alpha} d(s_i, r_i) \right) \enspace.
\]
By definition of $\L'$, we have $d(s, r_{i+1}) \leq \beta_{i+1}^{1 / \alpha} d(s_{i+1}, r_{i+1})$ and $d(s_{i+1}, r) \leq \beta_{i+1}^{1 / \alpha} d(s_{i+1}, r_{i+1})$. Thus, we get 
\[
d(s_i, r_{i + 1}) d(s_{i + 1}, r_i) \leq d(s, r_{i+1}) d(s_{i+1}, r) + 6 \beta_i^{1 / \alpha} \beta_{i + 1}^{1 / \alpha} d(s_i, r_i) d(s_{i + 1}, r_{i + 1}) \enspace.
\]
Since for all $x, y \geq 0$, we have $(x+y)^\alpha \leq (2x)^\alpha + (2y)^\alpha$, this yields
\begin{align*}
& d(s_i, r_{i + 1})^\alpha d(s_{i + 1}, r_i)^\alpha \\&\leq  2^\alpha d(s, r_{i+1}) d(s_{i+1}, r)^\alpha 
+ 12^\alpha \beta_i \beta_{i + 1} d(s_i, r_i)^\alpha d(s_{i + 1}, r_{i + 1})^\alpha \enspace.
\end{align*}
Using Equation~\eqref{eq:admissiblesets:newweights:sinr} yields
\begin{align*}
\beta_i \beta_{i+1} d(s_i, r_i)^\alpha d(s_{i + 1}, r_{i + 1})^\alpha \leq &\; \frac{2^\alpha}{2^\alpha(2 + 6^\alpha)} d(s, r_{i+1})^\alpha d(s_{i+1}, r)^\alpha \\
& + \frac{12^\alpha}{2^\alpha (2 + 6^\alpha)} \beta_i \beta_{i+1} d(s_i, r_i)^\alpha d(s_{i + 1}, r_{i + 1})^\alpha \enspace,
\end{align*}
or equivalently
\[
\beta_i \beta_{i+1} d(s_i, r_i)^\alpha d(s_{i + 1}, r_{i + 1})^\alpha \leq \frac{1}{2} d(s, r_{i+1})^\alpha d(s_{i+1}, r)^\alpha \enspace.
\]
This shows the first bound.

For the second bound, we use that both $d(s, r_{i+1})^\alpha$ and $d(s_{i+1}, r)^\alpha$ are at most $\beta_{i+1} d(s_{i+1}, r_{i+1})^\alpha$ by definition of $\L'$, to get from the previous inequality
\[
\beta_i \beta_{i+1} d(s_i, r_i)^\alpha d(s_{i + 1}, r_{i + 1})^\alpha \leq \frac{1}{2} \left( \beta_{i+1} d(s_{i+1}, r_{i+1})^\alpha \right)^2\enspace,
\]
and therefore
\[
\beta_i d(s_i, r_i)^\alpha \leq \frac{1}{2} \beta_{i+1} d(s_{i+1}, r_{i+1})^\alpha \enspace. 
\]
Applied inductively, we get that $\beta_{i+1} d(s_{i+1}, r_{i+1})^\alpha \geq 2^i \beta_1 d(s_1, r_1)^\alpha$. By definition, we have $\beta_1 d(s_1, r_1)^\alpha \geq z$, this shows the second bound.

Combining the two bounds, we get for all $i \geq 3$
\[
d(s, r_i)^\alpha d(s_i, r)^\alpha \geq 2 \beta_{i-1} \beta_i d(s_{i-1}, r_{i-1})^\alpha d(s_i, r_i)^\alpha \geq 2 \cdot 2^{i-2} z \beta_i d(s_i, r_i)^\alpha \enspace.
\]
This yields
\[
\sum_{i \in [\bar{n}]} \min \left\{ 1, z \cdot \frac{\beta_i \cdot d(s_i, r_i)^\alpha}{d(s, r_i)^\alpha \cdot d(s_i, r)^\alpha} \right\} \leq 2 + \sum_{i \geq 3} 2^{-(i - 1)} = \frac{5}{2} \enspace. 
\]
\end{proof}

\section{Proof of Lemma~\ref{lemma:feasiblesetcharacterization}}
In order to prove Lemma~\ref{lemma:feasiblesetcharacterization}, we show two separate claims. Both of them have been shown before in the context of constant thresholds (see \cite{Kesselheim2010,HalldorssonINFOCOM2012}). For the sake of completeness, we nevertheless present the full adapted versions here.

\begin{claim}
Given a set $\L \subseteq \R$ of links and a power assignment $p$ fulfilling $\frac{1}{\beta(\ell)} \frac{p(\ell)}{d(s,r)^\alpha} \leq \frac{1}{\beta(\ell')} \frac{p(\ell')}{d(s',r')^\alpha}$ for all $\ell, \ell' \in \L$ with $\pi(\ell) < \pi(\ell')$ such that for each link $\ell' \in \L$ the SINR is at least $\beta(\ell')$. Then we have for all links $\ell \in \R$
\[
\sum_{\substack{\ell' \in \L \\ \pi(\ell') < \pi(\ell)}} a_p(\ell', \ell) = O(1) \enspace.
\]
\end{claim}

\begin{proof}
We apply Lemma~\ref{lemma:admissibleset:geometrysenders} on the weights $w(\ell, \ell') = a_p(\ell, \ell')$.  Consider three links $\ell = (s,r), \ell'=(s',r'), \ell''=(s'',r'') \in \R$ such that $d(s', r) \leq 3 d(s'', r)$ and $\pi(\ell') < \pi(\ell'')$. For the power assignment, we have
\[
\frac{1}{\beta(\ell')} \frac{p(\ell')}{d(s',r')^\alpha} \leq \frac{1}{\beta(\ell'')} \frac{p(\ell'')}{d(s'',r'')^\alpha} \enspace.
\]
For the affectance this yields
\begin{align*}
a_p(\ell, \ell'') & = \min \left\{ 1, \beta(\ell'') \frac{p(\ell)}{d(s, r'')^\alpha} \Bigg/ \left( \frac{p(\ell'')}{d(s'', r'')^\alpha} - \beta(\ell'') \noise \right) \right\} \\
& \leq \min \left\{ 1, \beta(\ell') \frac{p(\ell)}{\frac{1}{3^\alpha} d(s, r')^\alpha} \Bigg/ \left( \frac{p(\ell'')}{d(s', r')^\alpha} - \beta(\ell') \noise \right) \right\} \\
& \leq 3^\alpha \cdot a_p(\ell, \ell') \enspace.
\end{align*}
That is, Lemma~\ref{lemma:admissibleset:geometrysenders} is applicable and yields the claim.
\end{proof}

For the second claim, we use essentially the same technique but this time combined with the notion of \emph{anti-feasibility} \cite{HalldorssonINFOCOM2012}. That is, we consider a set for which not only the incoming affectance is bounded by $1$ but also the outgoing affectance is bounded by $2$.

\begin{claim}
Given a set $\L \subseteq \R$ of links and a power assignment $p$ fulfilling $p(\ell') \geq p(\ell)$ if $\pi(\ell) < \pi(\ell')$ such that for each link $\ell' \in \L$ the SINR is at least $\beta(\ell')$. Then there is a subset $\L' \subseteq \L$ with $\lvert \L' \rvert = \Omega(\lvert \L \rvert)$ such that we have for all links $\ell \in \R$
\[
\sum_{\substack{\ell' \in \L' \\ \pi(\ell') < \pi(\ell)}} a_p(\ell, \ell') = O(1) \enspace.
\]
\end{claim}

\begin{proof}
Since $\L$ is feasible, we have $\sum_{\ell' \in \L} a_p(\ell', \ell) \leq 1$ for all $\ell \in \L$. Summing up these inequalities, we have $\sum_{\ell \in \L} \sum_{\ell' \in \L} a_p(\ell', \ell) \leq \lvert \L \rvert$. By Markov's inequality, this means that for at least half of the links $\ell'$ in $\L$, we have $\sum_{\ell \in \L} a_p(\ell', \ell) \leq 2$. Let $\L'$ be the set of these links. We now claim that this set fulfills the bound mentioned above.

We apply Lemma~\ref{lemma:admissibleset:geometryreceivers} on the weights $w(\ell, \ell') = a_p(\ell', \ell)$.  Consider three links $\ell = (s,r), \ell'=(s',r'), \ell''=(s'',r'') \in \R$ such that $d(s', r) \leq 3 d(s'', r)$ and $\pi(\ell') < \pi(\ell'')$. For the power assignment, we have $p(\ell') \geq p(\ell'')$ by this definition. For the weights this yields
\begin{align*}
w(\ell, \ell'') & = a_p(\ell'', \ell) \\
& = \min \left\{ 1, \beta(\ell) \frac{p(\ell'')}{d(s'', r)^\alpha} \Bigg/ \left( \frac{p(\ell)}{d(s, r)^\alpha} - \beta(\ell) \noise \right) \right\} \\
& \leq \min \left\{ 1, \beta(\ell) \frac{p(\ell')}{\frac{1}{3^\alpha} d(s', r)^\alpha} \Bigg/ \left( \frac{p(\ell)}{d(s, r)^\alpha} - \beta(\ell) \noise \right) \right\} \\
& \leq 3^\alpha \cdot a_p(\ell', \ell) = 3^\alpha \cdot w(\ell, \ell') \enspace.
\end{align*}
So, we can apply Lemma~\ref{lemma:admissibleset:geometryreceivers} on the set $\L'$ and the weights $w$.
\end{proof}

%% file: preliminaries.tex
\section{Technical Preliminaries}
In order to prove Lemma~\ref{lemma:admissiblesetcharacterization}, we need two prerequisites that are known techniques in the case of constant thresholds. for the sake of completeness, we present both of them briefly in the context of individual thresholds.

First, we consider the technique of signal strengthening, introduced by Halldorsson and Wattenhofer~\cite{Halldorsson2009}. Here, the key idea is that a set of links in which the SINR is at least $\beta$ can be decomposed into $\lceil 2 \cdot \beta' / \beta \rceil^2$ sets for any $\beta' \geq \beta$ such that in the resulting subsets the SINR is at least $\beta'$ for each link. Having individual thresholds for the links, the general technique still works in this case, now scaling all thresholds by some factor $c \geq 1$.

\begin{lemma}
Given a set of links $\L$ that is admissible with respect to SINR thresholds $\beta(\ell)$. For each $c \geq 1$, there is a decomposition of $\L$ into $4 \cdot c^2$ admissible sets with respect to SINR thresholds $\beta'(\ell) := c \cdot \beta(\ell)$. 
\end{lemma}

\begin{proof}
Let $p$ be the power assignment yielding SINR values of at least $\beta(\ell)$. We now consider the power assignment $2 \cdot c \cdot p$. This essentially scales the ambient noise by a factor of $2 \cdot c$. 

The decomposition into sets can be derived in two stages as follows. First, we iterate over $\L$ in an arbitrary order and put each link into the first set in which the SINR considering the previously added links is at least $2 \cdot c \cdot \beta(\ell)$. Note that this way at most $2 \cdot c$ sets are generated. In the second stage, each of the previously generated sets is treated the same way as the set $\L$ before in reverse order. Each of these sets is partitioned into at most $2 \cdot c$ subsets, making the total number of sets in the end at most $4 \cdot c^2$. In each of the sets, the SINR of each link when using power assignment $2 \cdot c \cdot p$ is at least $c \cdot \beta(\ell)$. 
\end{proof}

\begin{corollary}
\label{corollary:scalethresholdssubset}
Given a set of links $\L$ that is admissible with respect to SINR thresholds $\beta(\ell)$. For each $c \geq 1$, there is a subset $\L' \subseteq \L$ with $\lvert \L' \rvert \geq \lvert \L \rvert / 4 \cdot c^2$ that is admissible with respect to SINR thresholds $\beta'(\ell) := c \cdot \beta(\ell)$.
\end{corollary}

Furthermore, we consider the link set arising be swapping senders and receivers as already considered in \cite{Kesselheim2010}. One can show that in this instance of reversed links, the optimum cannot be significantly worse. 

\begin{lemma}
\label{lemma:dualadmissible}
For each admissible set $\L$ there is a subset $\L' \subseteq \L$ with $\lvert \L' \rvert = \Omega(\lvert \L \rvert)$ such that the set of reversed links $(\L')^\ast = \{ (r, s) \mid (s, r) \in \L' \}$ is admissible as well.
\end{lemma}

\begin{proof}
Let $p$ be the power assignment that makes $\L$ admissible. That is, for each $\ell \in \L$, we have
\[
\beta(\ell) \sum_{\substack{\ell' = (s', r') \in \L \\ \ell' \neq \ell}} \frac{p(\ell')}{d(s', r)^\alpha} \leq \frac{p(\ell)}{d(s, r)^\alpha} \enspace.
\]
For each $\ell = (s, r) \in \L$ let $p^\ast(\ell) = \beta(\ell) \frac{d(s, r)^\alpha}{p(\ell)}$. With this definition, we have
\[
\sum_{\substack{\ell' = (s', r') \in \L \\ \ell' \neq \ell}} \beta(\ell') \frac{d(s', r')^\alpha}{p^\ast(\ell')} \frac{1}{d(s', r)^\alpha} \leq \frac{1}{p^\ast(\ell)} \enspace,
\]
or equivalently
\[
\sum_{\substack{\ell' = (s', r') \in \L \\ \ell' \neq \ell}} \beta(\ell') \frac{d(s', r')^\alpha}{p^\ast(\ell')} \frac{p^\ast(\ell)}{d(s', r)^\alpha} \leq 1 \enspace.
\]
Taking the sum over all $\ell \in \L$, this is
\[
\sum_{\ell' = (s', r') \in \L} \beta(\ell') \sum_{\substack{\ell = (s, r) \in \L \\ \ell \neq \ell'}} \frac{d(s', r')^\alpha}{p^\ast(\ell')} \frac{p^\ast(\ell)}{d(s', r)^\alpha} \leq \lvert \L \rvert \enspace.
\]
That is, for at least $\lvert \L \rvert / 2$ links $\ell'$, we have
\[
\beta(\ell') \sum_{\substack{\ell = (s, r) \in \L \\ \ell \neq \ell'}} \frac{p^\ast(\ell)}{d(s', r)^\alpha} \leq 2 \frac{p^\ast(\ell')}{d(s', r')^\alpha} \enspace.
\]
Let the set of these links be denoted by $\L''$. By suitably scaling $p^\ast$, we furthermore have
\[
\beta(\ell') \sum_{\substack{\ell = (s, r) \in \L \\ \ell \neq \ell'}} \frac{p^\ast(\ell)}{d(s', r)^\alpha} + \noise \leq 3 \frac{p^\ast(\ell')}{d(s', r')^\alpha} \enspace.
\]
That is, $(\L'')^\ast$ is admissible with respect to thresholds $\frac{1}{3} \beta(\ell)$. Using Corollary~\ref{corollary:scalethresholdssubset}, we get the set $\L'$.
\end{proof}